\documentclass[12pt,draftclsnofoot]{IEEEtran}
\onecolumn
\usepackage{amsmath, amssymb,amsthm,cite}
\usepackage[dvips]{graphics}
\usepackage{graphicx}
\usepackage{psfrag}
\usepackage{dblfloatfix}
\usepackage{mathtools}
\usepackage{algpseudocode}
\usepackage{algorithm}
\usepackage{slashbox}
\DeclareMathOperator*{\argmax}{\arg\max}
\DeclarePairedDelimiter{\floor}{\lfloor}{\rfloor}
\allowdisplaybreaks
\newtheorem{lemma}{Lemma}
\newtheorem{theorem}{Theorem}

\theoremstyle{definition}

\newtheorem{definition}{Definition}
\makeatletter
\def\blfootnote{\gdef\@thefnmark{}\@footnotetext}
\makeatother
\begin{document}
\title{The Feedback Capacity of the $(1,\infty)$-RLL Input-Constrained Erasure Channel}
\author{Oron Sabag, Haim H. Permuter and Navin Kashyap}
\maketitle
\begin{abstract}
The input-constrained erasure channel with feedback is considered, where the binary input sequence contains no consecutive ones, i.e., it satisfies the $(1,\infty)$-RLL constraint. We derive the capacity for this setting, which can be expressed as $C_{\epsilon}=\max_{0 \leq p \leq \frac{1}{2}}\frac{H_{b}(p)}{p+\frac{1}{1-\epsilon}}$, where $\epsilon$ is the erasure probability and $ H_{b}(\cdot)$ is the binary entropy function. Moreover, we prove that a-priori knowledge of the erasure at the encoder does not increase the feedback capacity. The feedback capacity was calculated using an equivalent dynamic programming (DP) formulation with an optimal average-reward that is equal to the capacity. Furthermore, we obtained an optimal encoding procedure from the solution of the DP, leading to a capacity-achieving, zero-error coding scheme for our setting. DP is thus shown to be a tool not only for solving optimization problems such as capacity calculation, but also for constructing optimal coding schemes. The derived capacity expression also serves as the only non-trivial upper bound known on the capacity of the input-constrained erasure channel without feedback, a problem that is still open.
\end{abstract}
\section{Introduction}
\blfootnote{Part of this work will be presented at the 2015 Information Theory Workshop (ITW 2015), Jerusalem, Israel. The work of O. Sabag and H. H. Permuter was partially supported by the European Research Council (ERC) starting grant. All authors have also been partially supported by a Joint UGC-ISF research grant. O. Sabag and H. H. Permuter are with the department of Electrical and Computer Engineering, Ben-Gurion University of the Negev, Beer-Sheva, Israel (oronsa@post.bgu.ac.il, haimp@bgu.ac.il). N. Kashyap is with the department of Electrical Communication Engineering, Indian Institute of Science, Bangalore, India (nkashyap@ece.iisc.ernet.in).}
Memoryless channels have been the focus of research activity in information theory since they were introduced in 1948 by Shannon \cite{Shannon48}. The capacity of a memoryless channel has an elegant, single-letter expression, $C = \sup_{p(x)} I(X;Y)$, and this can be calculated for a broad range of channels \cite{Blahut72,Arimoto72}. When considering a memoryless channel with input that is constrained, the capacity is given by the maximum mutual information rate between the input and output sequences. The capacity calculation of such channels involves a calculation of the entropy rate of a Hidden Markov Model (HMM), since the transmission of a constrained sequence through a memoryless channel results in an output sequence that is described by an HMM. This makes the capacity of input-constrained memoryless channels difficult to compute \cite{vontobel_generalization,han_constrained_BSC_BEC,wolf_RLL,han_RLL_BSC}.
\begin{figure}[t]
\centering
    \psfrag{A}[b][][.8]{Constrained}
    \psfrag{K}[][][.8]{Encoder}
    \psfrag{B}[][][1]{$P_{Y|X}$}
    \psfrag{C}[][][1]{Decoder}
    \psfrag{D}[][][0.8]{Unit Delay}
    \psfrag{E}[][][1]{$M\in 2^{nR}$}
    \psfrag{F}[][][.9]{$x_i(m,y^{i-1})$}
    \psfrag{G}[][][1]{$y_i$}
    \psfrag{H}[][][1]{$y_i$}
    \psfrag{I}[][][1]{$y_{i-1}$}
    \psfrag{J}[][][1]{$\hat{M}(Y^n)$}
    \includegraphics[scale = 0.7]{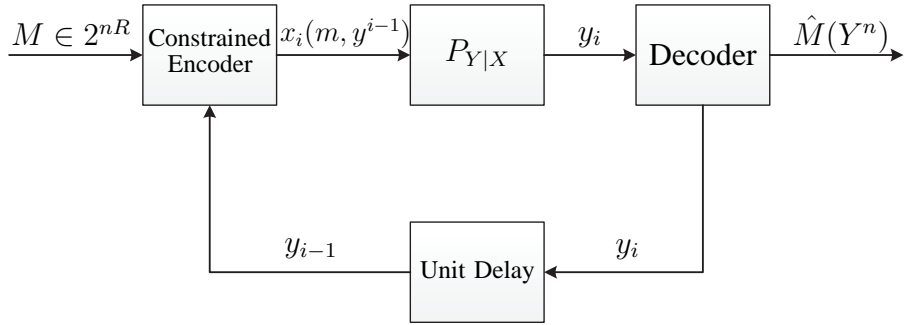}
    \caption{System model for an input-constrained memoryless channel with perfect feedback.}
    \label{fig:channel}
\end{figure}

Constrained coding arises naturally in many communication and recording systems\cite{Marcus98,Immink04}; a common constraint that is useful in magnetic and optical recording is the $(d,k)$-runlength limited (RLL) constraint. A binary sequence satisfies this constraint if the number of zeros between any pair of successive ones is at least $d$ and at most $k$. This constraint has also recently appeared in code designs for energy harvesting systems, where communication is used not only for information transfer but also for charging the receiver's battery\cite{osvalso_charge_battery}. In this paper, we focus on the special case of the $(1,\infty)$-RLL constraint, in which no consecutive ones are allowed.

It is well known that feedback does not increase the capacity of a memoryless channel, as shown by Shannon \cite{shannon56}. However, Shannon's argument does not apply to memoryless channels with constrained inputs, and special tools are required to determine the capacity of such channels with or without feedback.

We consider an $(1,\infty)$-RLL input-constrained binary erasure channel (BEC) with feedback, represented pictorially in Fig. \ref{fig:channel}, with the channel depicted in Fig. \ref{fig:erasure}. Based on the message $M$ and the previous channel outputs, $y^{i-1}$, the encoder chooses the input $X_{i}$, such that the input constraint is satisfied. The mechanism of the BEC is simple: each transmitted bit is transformed into an erasure symbol with probability $\epsilon$ or received successfully with its complementary probability. The decoder estimates the message $\hat{M}$ with low probability of error as a function of the output sequence $Y^{n}$. In this paper, we derive the explicit expression for the feedback capacity of the $(1,\infty)$-RLL input-constrained BEC.
\begin{figure}[t]
\centering
    \psfrag{A}[][][1]{$0$}
    \psfrag{B}[][][1]{$1$}
    \psfrag{C}[][][1]{$0$}
    \psfrag{D}[][][1]{$?$}
    \psfrag{E}[][][1]{$1$}
    \psfrag{F}[][][1]{$\epsilon$}
    \psfrag{G}[][][1]{$1-\epsilon$}
    \psfrag{X}[][][1]{$X$}
    \psfrag{Y}[][][1]{$Y$}
    \includegraphics[scale = 0.5]{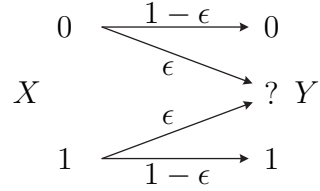}
    \caption{Erasure channel with erasure probability $\epsilon$.}
    \label{fig:erasure}
\end{figure}

The feedback capacity that is derived here also serves as an upper bound on the capacity of the $(1,\infty)$-RLL input-constrained BEC without feedback, a problem that is still open. A lower bound on the capacity of the non-feedback setting was derived in \cite{constrained_erasure_nofeedback_isit} by considering an input that is restricted to first-order Markov process (first-order capacity). The lower bound in \cite{constrained_erasure_nofeedback_isit} and our feedback capacity are presented in Fig. \ref{fig:comparison}, and it can be seen that maximal gap is attained at $\epsilon=0.71$, where the first-order capacity is $\sim 0.2354$ while the feedback capacity is $\sim 0.2547$.
\begin{figure}[h]
\centering
    \psfrag{A}[t][][.8]{Erasure probability $\epsilon$}
    \psfrag{B}[t][][1]{}
    \psfrag{C}[b][][1]{}
    \psfrag{E}[l][][.8]{Feedback Capacity}
    \psfrag{D}[r][][.8]{First-order capacity \cite{constrained_erasure_nofeedback_isit}}
    \psfrag{q}[][][.4]{$\mathbf{0}$}
    \psfrag{w}[][][.4]{$\mathbf{0.1}$}
    \psfrag{e}[][][.4]{$\mathbf{0.2}$}
    \psfrag{r}[][][.4]{$\mathbf{0.3}$}
    \psfrag{t}[][][.4]{$\mathbf{0.4}$}
    \psfrag{y}[][][.4]{$\mathbf{0.5}$}
    \psfrag{u}[][][.4]{$\mathbf{0.6}$}
    \psfrag{i}[][][.4]{$\mathbf{0.7}$}
    \psfrag{o}[][][.4]{$\mathbf{0.8}$}
    \psfrag{p}[][][.4]{$\mathbf{0.9}$}
    \psfrag{l}[][][.4]{$\mathbf{1}$}
    \psfrag{a}[r][][.4]{$\mathbf{0}$}
    \psfrag{s}[r][][.4]{$\mathbf{0.1}$}
    \psfrag{d}[r][][.4]{$\mathbf{0.2}$}
    \psfrag{f}[r][][.4]{$\mathbf{0.3}$}
    \psfrag{g}[r][][.4]{$\mathbf{0.4}$}
    \psfrag{h}[r][][.4]{$\mathbf{0.5}$}
    \psfrag{j}[r][][.4]{$\mathbf{0.6}$}
    \psfrag{k}[r][][.4]{$\mathbf{0.7}$}
    \includegraphics[scale = 0.32]{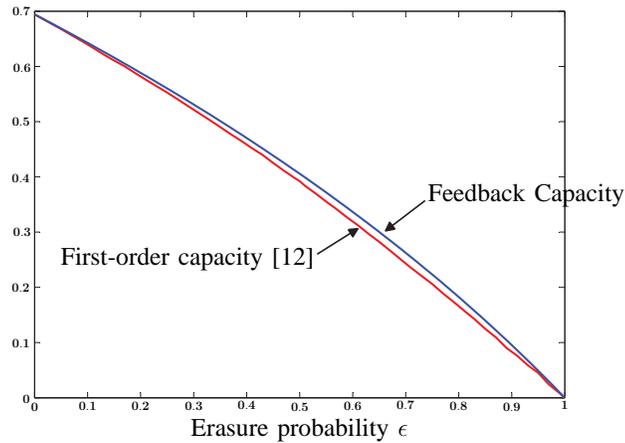}
    \caption{Lower and upper bounds on the capacity of the input-constrained BEC without feedback.}
    \label{fig:comparison}
\end{figure}

The relation between feedback-capacity calculation and dynamic programming (DP) first appeared in Tatikonda's thesis \cite{Tatikonda00}. Subsequent works included the formulation of capacity as DP for channels where the state is a function of the input \cite{Yang05}, Markov channels \cite{TatikondaMitter_IT09} and power-constrained Gaussian noise channels with memory \cite{YangKavcicTatikondaGaussian}. To apply algorithms from DP, such as value and policy iteration, quantization is required, and therefore, only lower bounds were derived in the above papers.

In \cite{PermuterCuffVanRoyWeissman08} and \cite{Ising_channel}, the feedback-capacities of the trapdoor and Ising channels, respectively, were found by solving their corresponding Bellman equations. The idea is that the feedback capacity is equal to the optimal reward of the DP, and therefore, it suffices to find a solution which satisfies the Bellman equation\cite{Bertsekas05}. Besides reward optimality verification, the Bellman equation also establishes a mechanism for optimal policy verification, which is a significant additional benefit.

The novelty in our work is the derivation of the optimal input distribution from the Bellman equation solution. The optimal solution of the DP is then utilized to understand how the dynamic program evolves under an optimal policy. We show that converting the DP solution into channel coding terms results in a straightforward interpretation of optimal encoding procedure. This encoding procedure led us to an innovative and zero-error coding scheme for our input-constrained setting. This establishes that DP as a tool is good not only for solving optimization problems, but also for deriving optimal coding schemes.

We also consider an input-constrained BEC where the encoder knows ahead of time if there is an erasure in the channel. Clearly, this non-causal setting is superior in terms of capacity compared to the feedback setting. We have managed to show that the capacity of this setting coincides with our feedback capacity expression, and therefore, a priori knowledge of the erasure in the channel does not increase the feedback capacity. Although this finding and the coding scheme for the feedback setting are sufficient for the feedback-capacity derivation, we argue that the capacity-achieving coding scheme is hard to construct without the DP solution.

The remainder of the paper is organized as follows. Section \ref{sec:definition} includes notation and description of the problem. Section \ref{sec:main} states the main results of this paper. In Section \ref{sec:formulation}, we provide a brief review of infinite-horizon DP and present the DP formulation of the feedback capacity. In Section \ref{sec:solution_for_eras}, the DP for the erasure channel is calculated, evaluated numerically and, finally, we prove that the Bellman equation is satisfied. In Section \ref{sec:relation_coding}, we present the derivation of the optimal scheme from the solution of the DP. In Section \ref{sec:upperbound}, we derive the capacity of non-causal input-constrained BEC. Finally, the paper is concluded in Section \ref{sec:conclusions}.

\section{Notation and Problem Definition}\label{sec:definition}
Throughout this paper, random variables will be denoted by upper-case letters, such as $X$, while realizations or specific values will be denoted by lower-case letters, e.g., $x$. Calligraphic letters will denote the alphabets of the random variables, e.g., $\mathcal{X}$.
Let $X^{n}$ denote the $n$-tuple $(X_{1},\dots,X_{n})$. For any scalar $\alpha\in[0,1]$, $\bar{\alpha}$ stands for $\bar{\alpha}=1-\alpha$. Let $H_b(\alpha)$ denote the binary entropy for scalar $\alpha\in[0,1]$, i.e., $H_b(\alpha)=-\alpha\log_2\alpha-\bar{\alpha}\log_2\bar{\alpha}$. Let $H_{ter}(\alpha_1,\alpha_2,\alpha_3)$ denote the ternary entropy for scalars $\alpha_1,\alpha_2,\alpha_3\in[0,1]$ such that $\sum_i\alpha_i = 1$, i.e., $H_{ter}(\alpha_1,\alpha_2,\alpha_3)=\sum_i-\alpha_i\log_2\alpha_i$.

The communication setting of a memoryless channel with feedback is described in Fig. \ref{fig:channel}. A message $M$ is drawn uniformly from the set $\{1,\dots,2^{nR}\}$ and made available to the encoder. The encoder at time $i$ knows the message $m$ and the feedback samples $y^{i-1}$, and produces a binary output, $x_i \in \{0,1\}$, as a function of $m$ and $y^{i-1}$. The sequence of encoder outputs, $x_1x_2x_3\ldots$, must satisfy the $(1,\infty)$-RLL input-constraint of the channel, namely, no two consecutive ones are allowed. The channel is memoryless in the sense that the output at time $i$, given the existing information in the system, depends only on the current input, i.e.,
\begin{align}\label{def_memoryless}
    p(y_i|x^{i},y^{i-1})=p(y_i|x_i),\ \forall i.
\end{align}

We focus on the erasure channel, shown in Fig. \ref{fig:erasure}. The input alphabet is $\mathcal{X}=\{0,1\}$, while the output can take values in $\mathcal{Y}=\{0,1,?\}$. The probability for erasure in the channel is $\epsilon$ and can take any value in $[0,1]$.
\begin{definition} A $(n,2^{nR},(1,\infty))$ \emph{code} for a constrained-input channel with feedback is defined by a set of encoding functions:
    \begin{equation*}
        f_i: \{1,\dots,2^{nR}\}\times \mathcal{Y}^{i-1} \rightarrow \mathcal{X}_i, \ i=1,\dots,n,
    \end{equation*}
    satisfying $f_i(m,y^{i-1})=0 \  \text{if} \ f_{i-1}(m,y^{i-2})=1$ for all $(m,y^{i-1})$, and a decoding function:
    \begin{equation*}
        \Psi: \mathcal{Y}^{n} \rightarrow \{1,\dots,2^{nR}\}.
    \end{equation*}
\end{definition}

In addition, we define the non-causal $(1,\infty)$-RLL BEC. For this setting, all definitions remain the same as in the previous setting, but the encoder knows ahead of time whether there is an erasure in the channel. Formally, define $\theta_i$ as the indicator that corresponds to erasure in the channel at time $i$, namely, $\theta_i=0$ if $x_i=y_i$ and $\theta_i=1$ otherwise. The set of encoding functions for this setup is then defined as:
\begin{equation*}
    f_i: \{1,\dots,2^{nR}\}\times \mathcal{Y}^{i-1}\times\{0,1\} \rightarrow \mathcal{X}_i, \ i=1,\dots,n,
\end{equation*}
satisfying $f_i(m,y^{i-1},\theta_i)=0 \  \text{if} \ f_{i-1}(m,y^{i-2},\theta_{i-1})=1$ for all $(m,y^{i-1},\theta_{i-1},\theta_{i})$.

The \textit{average probability of error} for a code is defined as $P_{e}^{(n)}=\Pr(M\neq\Psi(Y^{n}))$. A rate $R$ is said to be $(1,\infty)$\textit{-achievable} if there exists a sequence of $(n,2^{nR},(1,\infty))$ codes, such that $\lim_{n\rightarrow\infty} P_{e}^{(n)}=0$. The \textit{capacity}, $C^{\text{fb}}_\epsilon$, defined to be the supremum over all $(1,\infty)$-achievable rates, is a function of the erasure probability $\epsilon$. Let $C^{\text{nc}}_\epsilon$ denote the capacity for the non-causal $(1,\infty)$-RLL BEC. From operational considerations of the encoding functions for both settings, it is clear that $C^{\text{nc}}_\epsilon\geq C^{\text{fb}}_\epsilon$.
\section{Main Results}\label{sec:main}
The following is our main result concerning the capacity of the $(1,\infty)$-RLL constrained BEC with feedback.
\begin{theorem}\label{theorem:main}
The capacity of the $(1,\infty)$-RLL input-constrained erasure channel with feedback is
\begin{align}\label{eq:main_capacity}
    C^{\text{fb}}_{\epsilon} &= \max_{0\leq p \leq \frac{1}{2}} \frac{H_b(p)}{p+\frac{1}{1-\epsilon}}.
\end{align}
Furthermore, the capacity is achieved by an explicit zero-error coding scheme that is presented in Section \ref{subsec:coding_scheme}, in Algorithm \ref{alg:encoding} and Algorithm \ref{alg:decoding}.
\end{theorem}
In Fig. \ref{fig:capacity}, the feedback capacity is evaluated for different values of erasure probability $\epsilon$. As can be seen, the capacity is a decreasing function for an increasing value of $\epsilon$. For $\epsilon=0$, the capacity is $C^{\text{fb}}_{0}\approx 0.6942$, which can be represented as $\log_2 \phi$, where $\phi$ is the golden ratio and is known as the entropy rate of a binary source with no consecutive ones. For $\epsilon=1$, the capacity value is $C^{\text{fb}}_{1}=0$, as expected.

The capacity of the non-constrained BEC can be expressed as $\max_{0\leq p\leq\frac{1}{2}}\frac{H_{b}(p)}{\frac{1}{1-\epsilon}}=1-\epsilon$. Note that the only difference between this term and our capacity expression in \eqref{eq:main_capacity} is the denominator. This fact hints that the capacity expressions of other input constraints may share a common structure.
\begin{figure}[!t]
\centering
    \psfrag{A}[][][1]{}
    \psfrag{B}[t][][1]{Erasure probability $\epsilon$}
    \psfrag{C}[b][][1]{Capacity $C^{\text{fb}}_{\epsilon}$}
    \includegraphics[scale = 0.28]{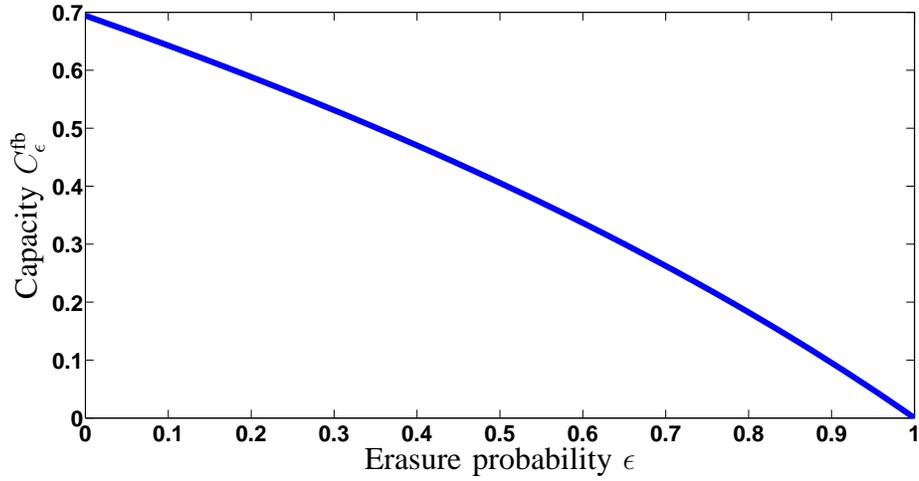}
    \caption{The capacity $C^{\text{fb}}_{\epsilon}$, as a function of $\epsilon$, of the $(1,\infty)$-RLL input-constrained BEC with feedback.}
    \label{fig:capacity}
\end{figure}

The next theorem states that the non-causal $(1,\infty)$-RLL input-constrained BEC has the same capacity as the feedback setting.
\begin{theorem}\label{theorem:non_causal}
Non-causal knowledge of erasures does not increase the feedback capacity, i.e.,
\begin{align*}
C^{\text{nc}}_{\epsilon}&=C^{\text{fb}}_{\epsilon}.
\end{align*}
\end{theorem}
Next, we show the properties of the capacity expression \eqref{eq:main_capacity}.
\begin{lemma}\label{lemma:concave} Define the function $f_{\epsilon}(p)=\frac{H_b(p)}{p+\frac{1}{1-\epsilon}}$, where $p\in[0,1]$. The following properties hold for $f_{\epsilon}(p)$:
\begin{itemize}
\item The function $f_{\epsilon}(p)$ is concave on $[0,1]$, for any $\epsilon\geq 0$.
\item The function $f_{\epsilon}(p)$ has only one maximum in $[0,1]$, which is the only real solution of the equation $p^{\frac{1}{\bar{\epsilon}}}=(1-p)^{1+\frac{1}{\bar{\epsilon}}}$. This maximum lies in $[0,\frac{1}{2}]$.
\item Denote by $p_{\epsilon}$ the argument that achieves the maximum of $f_{\epsilon}(p)$. The capacity can also be expressed by,
\begin{equation*}
C_{\epsilon}^{\text{fb}} = \frac{-\log_2(p_{\epsilon})}{1+\frac{1}{1-\epsilon}}.
\end{equation*}
\end{itemize}
\end{lemma}
The proof of Lemma \ref{lemma:concave} is presented in Appendix \ref{app:lemma concave}.
\section{Feedback Capacity and Dynamic Programming}\label{sec:formulation}
The normalized, directed information was introduced by Massey in \cite{Massey90} as $\frac{1}{n}I(X^n\rightarrow Y^n)= \frac{1}{n}\sum_{i=1}^n I(X^i;Y_i|Y^{i-1})$. Massey showed that the maximum normalized directed information upper bounds the capacity of channels with feedback, and subsequently, it was proved that this expression indeed characterizes the feedback capacity for a broad class of channels \cite{TatikondaMitter_IT09,Kramer03,Kim08_feedback_directed,PermuterWeissmanGoldsmith09,ShraderPermuter09CompoundIT}. Of most relevance to our work is the feedback capacity of the unifilar finite state channel that was characterized in \cite{PermuterCuffVanRoyWeissman08}. The next theorem follows from Theorem $1$ in \cite{PermuterCuffVanRoyWeissman08}, by substituting $S_{t-1}=X_{t-1}$ as the channel state at time $t$.

\begin{theorem}[Theorem 1, \cite{PermuterCuffVanRoyWeissman08}]\label{theorem:cap_as_DP}
The capacity of an $(1,\infty)$-RLL input-constrained memoryless channel with feedback can be written as:
\begin{equation}\label{eq_capacity_as_DP}
  C_{\epsilon}^{\text{fb}} = \sup \liminf_{N\rightarrow \infty} \frac{1}{N} \sum_{t=1}^{N}I(X_t,X_{t-1};Y_t|Y^{t-1}),
\end{equation}
where the supremum is taken with respect to $\{p(x_{t}|x_{t-1},y^{t-1}):p(x_{t}=1|x_{t-1}=1,y^{t-1})=0\}_{t\geq 1}$.
\end{theorem}
Having written the capacity of the input constrained channel with feedback as \eqref{eq_capacity_as_DP}, we proceed to show that calculating the capacity can be formulated as an average-reward DP.
\subsection{Average-Reward Dynamic Programs}\label{subsec:Av_Re_DP}
Each DP is defined by the tuple $(\mathcal{Z},\mathcal{U},\mathcal{W},F,P_Z,P_w,g)$.
We consider a discrete-time dynamic system evolving according to:
\begin{equation}\label{eq_DP}
  z_t=F(z_{t-1},u_t,w_t),\  t=1,2,\dots
\end{equation}
Each state, $z_t$, takes values in a Borel space $\mathcal{Z}$, each action, $u_t$, takes values in a compact subset $\mathcal{U}$ of a Borel space, and each disturbance, $w_t$, takes values in a measurable space $\mathcal{W}$. The initial state, $z_0$, is drawn from the distribution $P_Z$, and the disturbance, $w_t$, is drawn from $P_w(\cdot|z_{t-1},u_t)$.
The history, $h_t=(z_0,w_1,\dots,w_{t-1})$, summarizes all the information available to the controller at time $t$. The controller at time $t$ chooses the action, $u_t$, by a function $\mu_t$ that maps histories to actions, i.e., $u_t = \mu_t(h_t)$. The collection of these functions is called a policy and is denoted as $\pi=\{\mu_1,\mu_2,\dots\}$. Note that given a policy, $\pi$, and the history, $h_t$, one can compute the actions vector, $u^t$, and the states of the system, $z_1,z_2,\dots,z_{t-1}$.

Our objective is to maximize the average reward given a bounded reward function $g: \mathcal{Z}\times \mathcal{U}\rightarrow \mathbb{R}$. The average reward for a given policy $\pi$ is given by:
\begin{equation*}
\rho_{\pi} = \liminf _{N\rightarrow\infty} \frac{1}{N}\mathbb{E}_{\pi}\left[\sum_{t=1}^{N}g(Z_{t-1}, \mu_t(h_t))\right],
\end{equation*}
where the subscript $\pi$ indicates that actions $u_t$ are generated by the policy $\pi$. The optimal average reward is defined as
\begin{equation*}
\rho = \sup_{\pi} \rho_{\pi}.
\end{equation*}

\subsection{Formulation of the feedback capacity as DP}\label{subsec:formulation}
The state of the dynamic programming, $z_{t-1}$, is defined as the conditioned probability vector $\beta_{t-1}(x_{t-1}) = p(x_{t-1}|y^{t-1})$. The action space, $\mathcal{U}$, is the set of stochastic matrices, $p(x_t|x_{t-1})$, satisfying the $(1,\infty)$-RLL constraint. For a given policy and an initial state, the encoder at time $t-1$ can calculate the state, $\beta_{t-1}(x_{t-1})$, since the tuple $y^{t-1}$ is available from the feedback. The disturbance is taken to be the channel output, $w_t=y_{t}$, and the reward gained at time $t-1$ is chosen as $I(Y_{t};X_{t},X_{t-1}|y^{t-1})$. The formulation is summarized in Table \ref{table:formulation}.

\textbf{Existence of System:}
We need to show that for a given policy, $\pi=\{\mu_1,\mu_2,\dots\}$, the state $z_t$ can be calculated from the tuple $(z_{t-1},u_t,y_{t})$.
Consider,
\begin{align}\label{state_DP}
\beta_t(x_t)&=p(x_t|y^{t}) \nonumber\\
&=\sum_{x_{t-1}} p(x_t,x_{t-1}|y^{t}) \nonumber\\
&=\frac{\sum_{x_{t-1}} p(x_t,x_{t-1},y_{t}|y^{t-1})}{p(y_{t}|y^{t-1})} \nonumber\\
&= \frac{\sum_{x_{t-1}} p(x_{t-1}|y^{t-1})p(x_t|x_{t-1},y^{t-1})p(y_{t}|y^{t-1},x_t,x_{t-1})}{\sum_{x_{t},x_{t-1}} p(y_{t},x_{t},x_{t-1}|y^{t-1})} \nonumber\\
&\stackrel{(a)}= \frac{\sum_{x_{t-1}} p(x_{t-1}|y^{t-1})p(x_t|x_{t-1},y^{t-1})p(y_{t}|x_t)}{\sum_{x_{t},x_{t-1}} p(x_{t-1}|y^{t-1})p(x_{t}|x_{t-1},y^{t-1})p(y_{t}|x_{t})} \nonumber\\
&= \frac{\sum_{x_{t-1}} \beta_{t-1}(x_{t-1})u_t(x_t,x_{t-1})p(y_{t}|x_t)}{\sum_{x_{t},x_{t-1}} \beta_{t-1}(x_{t-1})u_t(x_{t},x_{t-1})p(y_{t}|x_{t})},
\end{align}
where $(a)$ follows from the memoryless property \eqref{def_memoryless}.
Therefore, there exists a function $F$, such that $\beta_t(x_t)=F(\beta_{t-1}(x_{t-1}),u_t(x_t,x_{t-1}),w_{t})$.
\begin{table}[ht]
\caption{Formulation of capacity as DP}
\centering
\begin{tabular}{|c|c|}
  \hline
  Input-constrained memoryless channel & Dynamic Programming \\ \hline\hline
  $p(x_{t-1}|y^{t-1})$  & $z_{t-1}$, state at time $t-1$\\ \hline
  Constrained $p(x_t|x_{t-1})$ & $u_t$, action taken at time $t-1$  \\ \hline
  $y_{t}$ & $w_t$, disturbance generated at time $t$ \\ \hline
  Equation \eqref{state_DP} & $z_t=F(z_{t-1},u_t,w_t)$, system equation   \\ \hline
  $I(Y_{t};X_{t},X_{t-1}|y^{t-1})$ & $g(z_{t-1},u_t)$, reward gained at time $t-1$  \\ \hline
\end{tabular}\label{table:formulation}
\end{table}

\textbf{Disturbance:}
Let us show that the disturbance distribution depends on the current state and action only, with no dependence on past information, i.e., $p(w_t|w^{t-1},z^{t-1},u^t) = p(w_t|z_{t-1},u_t)$.
\begin{align*}
p(w_t|w^{t-1},z^{t-1},u^t)&=p(y_{t}|y^{t-1},\beta^{t-1},u^t)\\
&=\sum_{x_{t},x_{t-1}} p(y_{t},x_{t},x_{t-1}|y^{t-1},\beta^{t-1},u^t)\\
&=\sum_{x_{t},x_{t-1}} p(x_{t-1}|y^{t-1},\beta^{t-1},u^t)p(x_{t}|x_{t-1},y^{t-1},\beta^{t-1},u^t)p(y_{t}|x_{t},x_{t-1},\beta^{t-1},u^t,y^{t-1})\\
&\stackrel{(a)}=\sum_{x_{t},x_{t-1}} p(x_{t-1}|\beta_{t-1},u_t)p(x_{t}|x_{t-1},\beta_{t-1},u_t)p(y_{t}|x_{t})\\
&=\sum_{x_{t},x_{t-1}} p(y_{t},x_{t},x_{t-1}|\beta_{t-1},u_t)\\
&= p(y_{t}|\beta_{t-1},u_t)\\
&= p(w_{t}|z_{t-1},u_t),\\
\end{align*}
where $(a)$ follows from the fact that the value of $p(x_{t-1}|y^{t-1},\beta^{t-1},u^t)$ is determined by $\beta_{t-1}$, the fact that $x_t$ depends only on the triplet $(x_{t-1},\beta_{t-1},u_t)$, and finally, the fact that the channel is memoryless.

\textbf{Reward:}
We need to show that the reward, $I(Y_{t};X_{t},X_{t-1}|y^{t-1})$, that is achieved at time $t-1$ is a function of the current state, $\beta_{t-1}(x_{t-1})$, and of the chosen action $u_{t}$. Note that the term of the reward depends on the conditional distribution $p(y_t,x_{t},x_{t-1}|y^{t-1})$ only.

For an initial state $z_0$ and a given policy $\pi=\{\mu_1,\mu_2,\dots\}$, the term $\beta_{t-1}$ is determined by $y^{t-1}$. Let us show that the reward achieved at time $t-1$ depends on the current state, action and the channel characterization,
\begin{align*}
p(y_t,x_{t},x_{t-1}|y^{t-1})&\stackrel{(a)}=p(x_{t-1}|y^{t-1})p(x_{t}|x_{t-1},y^{t-1})p(y_{t}|x_{t})\\
&=\beta_{t-1}(x_{t-1})u_t(x_t,x_{t-1})p(y_{t}|x_{t}),
\end{align*}
where $(a)$ follows from the chain rule and the memoryless property \eqref{def_memoryless}. Recall that the term $p(y_t|x_t)$ is given by the channel characterization, and thus, the reward depends on the state, $\beta_{t-1}$, and the chosen action, $u_t$. Therefore, the reward at time $t-1$ can be written as:
\begin{equation*}
    g(z_{t-1},u_t) = I(Y_{t};X_t,X_{t-1}|\beta_{t-1},u_t).
\end{equation*}
It then follows that the optimal average reward of the DP is:
\begin{align*}
  \rho^{\ast} &=  \sup_{\pi}\liminf _{N\rightarrow\infty} \frac{1}{N}\sum_{t=1}^{N}I_{\pi}(Y_{t};X_t,X_{t-1}|Y^{t-1}),
\end{align*}
where the subscript $\pi$ indicates that the mutual information is calculated with respect to the policy $\pi$. This term is the capacity for an input-constrained memoryless channel with feedback as presented in Theorem \ref{theorem:cap_as_DP}, and we conclude that the optimal average reward is equal to the capacity.

\section{Solution For the Erasure Channel}\label{sec:solution_for_eras}
This section is organized as follows: Section \ref{subsec:form_for_erasure} formulates feedback capacity of the BEC as DP using the notation from Section \ref{subsec:formulation}. In Section \ref{subsec:numer_eval}, we evaluate a numerical solution using the value iteration algorithm, and finally, in Section \ref{subsec:bellman}, we present the Bellman equation and its solution for the BEC. The solution of the Bellman equation concludes the derivation of the feedback capacity expression in Theorem. \ref{theorem:main}.
\subsection{Formulation of the erasure channel as DP}\label{subsec:form_for_erasure}
The state of the DP at time $t-1$, $z_{t-1}$, is the probability vector $[p(x_{t-1}=0|y^{t-1}), p(x_{t-1}=1|y^{t-1})]$. With some abuse of notation, we refer from now on to $z_{t-1}\triangleq p(x_{t-1}=0|y^{t-1})$ as the first component of the vector, which also determines the second component, since they sum to $1$.
Each action, $u_{t}$, is a constrained $2\times2$ stochastic matrix, $p(x_t|x_{t-1})$, of the form:
\begin{equation*}
u_{t}=\left[\begin{array}{cc}
 p(x_t=0|x_{t-1}=0) & p(x_t=1|x_{t-1}=0) \\
1 & 0 \end{array}\right].
\end{equation*}

The disturbance $w_t$ is the channel output, $y_t$, and can take values in $\{0,1,?\}$. With the above definitions and \eqref{state_DP}, the system equation can be expressed as follows:
\begin{equation}\label{eq:system_erasure}
z_{t}=\left\{\begin{array}{cc}
 1 & \text{if } w_t=0, \\
 1 - z_{t-1} + z_{t-1}u_{t}(1,1) & \text{if } w_t=?, \\
0 & \text{if } w_t=1. \end{array}\right.
\end{equation}

At this point, to simplify notations we note that $1 - z_{t-1} + z_{t-1}u_{t}(1,1)$ can be written as $1-z_{t-1}u_{t}(1,2)$ . We denote $\delta_t\triangleq z_{t-1}u_{t}(1,2)$, and this implies the constraint $0\leq\delta_t\leq z_{t-1}$, since $u_t$, by definition, must be a stochastic matrix. Furthermore, when investigating the relation of DP and encoding procedures, $u_t$ has to be recovered from $\delta_t$, given $z_{t-1}$. This calculation is trivial for $z_{t-1}\neq0$, while for $z_{t-1}=0$, we note that $u_t(1,2)$ has no effect on the DP, and therefore, $u_t(1,2)$ can be fixed to zero.

To calculate the reward, the conditional distribution $p(x_t,x_{t-1},y_t|z_{t-1},u_t)$ is described in Table \ref{table:joint_distr}, and it follows that the reward is:
\begin{align*}
g(z_{t-1},u_t) &= I(Y_t;X_t,X_{t-1}|z_{t-1},u_t) \\
&= H(Y_t|z_{t-1},u_t) - H(Y_t|X_t,X_{t-1},z_{t-1},u_t)\\
&\stackrel{(a)}= H_{ter}((1-\delta_t)\bar{\epsilon}, \epsilon , \delta_t\bar{\epsilon}) - H_b(\epsilon)\\
&\stackrel{(b)}= H_{b}(\epsilon)+ \bar{\epsilon}H_b(\delta_t) - H_b(\epsilon)\\
&= \bar{\epsilon}H_b(\delta_t),
\end{align*}
where $(a)$ follows from the marginal distribution $p(y_t|z_{t-1},u_t)$ in Table \ref{table:joint_distr} and the definition of $\delta_t$, while $(b)$ follows from an easily verifiable identity: $H_{ter}(a\bar{b},\bar{a}\bar{b},b)=H_{b}(b) + \bar{b}H_{b}(a)$, for all $a,b\in [0,1]$.
\begin{table}[ht]
\caption{the conditional distribution $p(x_t,x_{t-1},y_t|z_{t-1},u_t)$}
\label{table:joint_distr}
\centering
\begin{tabular}{|c|c||c|c|c|}
  \hline
  $x_t$ & $x_{t-1}$ & $y_t=0$ & $y_t=?$ & $y_t=1$ \\ \hline \hline
  $0$ & $0$ & $z_{t-1}u_t(1,1)\bar{\epsilon}$ & $z_{t-1}u_t(1,1)\epsilon$ & $0$ \\ \hline
  $1$ & $0$ & $0$ & $z_{t-1}u_t(1,2)\epsilon$ & $z_{t-1}u_t(1,2)\bar{\epsilon}$ \\ \hline
  $0$ & $1$ & $(1-z_{t-1})\bar{\epsilon}$ & $(1-z_{t-1})\epsilon$ & $0$ \\ \hline
\end{tabular}
\end{table}

To apply the value iteration in the next subsection, it is convenient to define the operator of the DP:
\begin{align}\label{eq_operator}
    (Th)(z) &= \sup_{u\in\mathcal{U}} g(z,u) + \int P_{W}(dw|z,u)h(F(z,u,w)),
\end{align}
for all functions $h:\mathcal{Z}\rightarrow \mathbb{R}$.

For our case, the operator of the DP takes the form of
\begin{align}\label{eq_operator_erasure}
    (Th_{\epsilon})(z)&= \sup_{0\leq\delta \leq z} \bar{\epsilon}H_b(\delta) + (1-\delta)\bar{\epsilon}h_{\epsilon}(1) + \epsilon h_{\epsilon}(1-\delta) + \delta\bar{\epsilon}h_{\epsilon}(0),
\end{align}
for all $h_{\epsilon}:[0,1]\rightarrow \mathbb{R}$, where the subscript $\epsilon$ indicates that $h_{\epsilon}$ depends on the parameter ${\epsilon}$.
\subsection{Numerical evaluation}\label{subsec:numer_eval}
Now, that we have the DP formulation for our problem, we can apply the value iteration algorithm to estimate the optimal average reward. The value iteration algorithm is simply applying the DP operator from \eqref{eq_operator_erasure} successively, and it has the form $h_k(z)=(Th_{k-1})(z)$ with $h_{0}(z)=0$. The state of the DP and the values in the action matrices are continuous, which cannot be implemented by a finite-precision computer. To this end, a quantization of $5000$ points in the unit interval for both $z_t$ and $\delta_t$ was performed, and the results after 20 iterations are presented in Fig. \ref{fig:numerical} for erasure probability $\epsilon=0.5$.
\begin{figure}[h]
\centering
        \psfrag{A}[b][][1]{Action $\delta_{20}$}
        \psfrag{B}[b][][1]{Value function $h_{20}$}
        \psfrag{C}[t][][1]{State $z$}
        \psfrag{D}[b][][1]{$\delta_{20}$}
        \psfrag{E}[b][][1]{$h_{20}$}
        \centerline{\includegraphics[scale=0.4]{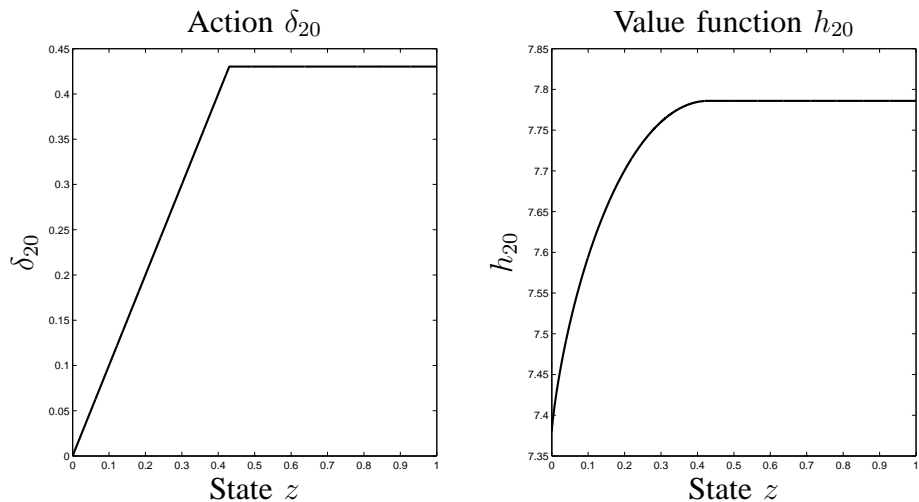}}
\caption{Value iteration evaluation for the erasure channel with $\epsilon=0.5$. The algorithm was implemented with 20 iterations and quantization of 5000 points for both action and state.}
\label{fig:numerical}
\end{figure}

We also simulated the system with the estimated optimal action $\delta_{20}$. The initial state, $z_0$, was chosen to be zero and the action was taken according to $\delta_{20}$ which led to a gained reward. The disturbance was generated randomly according to the induced distribution from Table \ref{table:joint_distr}. Having in hand the current state, action and disturbance, the new state was calculated and the process was repeated $10^6$ times. This simulation led to an approximate average reward of $0.4056$ and the histogram of the states is shown in Fig. \ref{fig:hist}. The significant importance of a discrete histogram will be discussed in Section \ref{sec:relation_coding}, where it is explained how the DP simulation leads us to derive an optimal coding scheme for our channel setting.
\begin{figure}[h]
\centering
        \psfrag{A}[b][][0.8]{Histogram of the state $z$}
        \centerline{\includegraphics[height=6cm]{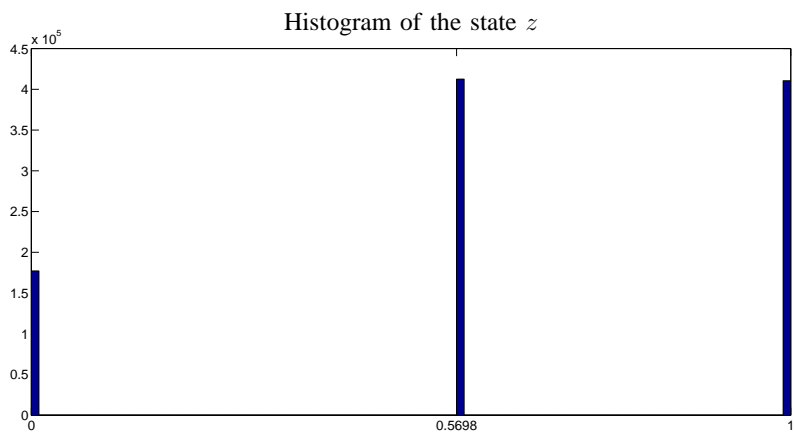}}
\caption{Histogram of system states after $10^6$ runs.}
\label{fig:hist}
\end{figure}
\subsection{The Bellman Equation}\label{subsec:bellman}
In dynamic programming, the Bellman equation suggests a sufficient condition for average reward optimality. This equation establishes a mechanism for verifying that a given average reward is optimal. The next result encapsulates the Bellman equation and can be found in \cite{Arapos93_average_cose_survey}.
\begin{theorem}[Theorem 6.1, \cite{Arapos93_average_cose_survey}]\label{theorem:bellman}
If $\rho\in\mathbb{R}$ and a bounded function $h:\mathcal{Z}\rightarrow \mathbb{R}$ satisfies for all $z \in \mathcal{Z}$:
\begin{equation}\label{eq_bellman}
    \rho + h(z) = \sup_{u\in\mathcal{U}} g(z,u) + \int P_{W}(dw|z,u)h(F(z,u,w)),
\end{equation}
then $\rho=\rho^{\ast}$. Furthermore, if there is a function $\mu : \mathcal{Z} \rightarrow \mathcal{U}$ such that $\mu(z)$ attains the supremum for each $z$, then $\rho_\pi = \rho^{\ast}$ for $\pi=\{\mu_0, \mu_1,\dots\}$ with $\mu_t(h_t) = \mu(z_{t-1})$ for each $t$.
\end{theorem}

For our DP, substituting \eqref{eq_operator_erasure} into \eqref{eq_bellman} yields the next Bellman equation:
\begin{align}\label{eq_our_bellman}
h_{\epsilon}(z) + \rho_{\epsilon} &= \sup_{0\leq\delta \leq z} \bar{\epsilon}H_{b}(\delta) + \bar{\epsilon}(1-\delta)h_{\epsilon}(1) + \epsilon h_{\epsilon}(1-\delta) + \bar{\epsilon}\delta h_{\epsilon}(0),
\end{align}
for all functions $h_{\epsilon}$.
Let us denote two constants $\rho_{\epsilon}^{\ast}$ and $p_{\epsilon}$,
\begin{align}\label{eq:rho_z_defintion}
   \rho_{\epsilon}^{\ast} &= \max_{0\leq p \leq \frac{1}{2}} \frac{H_b(p)}{p+\frac{1}{1-\epsilon}}, \nonumber\\
    p_\epsilon &= \argmax_{0\leq p \leq \frac{1}{2}} \frac{H_b(p)}{p+\frac{1}{1-\epsilon}},
\end{align}
and a bounded function,
\begin{align}\label{eq:func_star}
 h_{\epsilon}^{\ast}(z) =
 \begin{cases} \bar{\epsilon}H_b(z) - z\bar{\epsilon}\frac{H_b(p_\epsilon)}{p_\epsilon+\frac{1}{1-\epsilon}} &\mbox{if } 0 \leq z \leq p_{\epsilon} \\
        \frac{H_b(p_\epsilon)}{p_\epsilon+\frac{1}{1-\epsilon}} & \mbox{if } p_{\epsilon} \leq z \leq 1. \end{cases}
\end{align}
We proceed to show the DP solution by solving \eqref{eq_our_bellman}.
\begin{theorem}\label{theorem:sol_bellman}
The constant $\rho_{\epsilon}^{\ast}$ and the function $h_\epsilon^{\ast}(z)$ given in \eqref{eq:rho_z_defintion} and \eqref{eq:func_star}, respectively, satisfy the Bellman equation \eqref{eq_our_bellman} for each $\epsilon$.
Therefore, $\rho_{\epsilon}^{\ast}$ is the optimal average reward.
\end{theorem}
As $\rho_{\epsilon}^{\ast}$ is equal to the capacity expression \eqref{eq:main_capacity}, Theorem \ref{theorem:sol_bellman} concludes the proof for the first part of Theorem \ref{theorem:main}. The proof of Theorem \ref{theorem:sol_bellman} is presented in Appendix \ref{app:proof_bellman}.
\section{Derivation of the capacity-achieving coding scheme from the DP solution}\label{sec:relation_coding}
In this section, we derive the optimal coding scheme using the DP solution and finally show that this leads to a capacity-achieving coding scheme. The method comprises recovering the optimal constrained input distributions $\{p(x_t|x_{t-1},y^{t-1})\}_{t\geq 1}$ from the solution of the DP.
\subsection{Relation of the Coding Scheme to Dynamic Programming Results}\label{subsec:relation}

\begin{figure}[h!]
\centering
        \psfrag{A}[][][1.2]{$z=0$}
        \psfrag{B}[][][1.2]{$z=1$}
        \psfrag{C}[][][1]{$z=1-p$}
        \psfrag{D}[][][0.9]{$\delta^{\ast}=0$}
        \psfrag{E}[][][0.9]{$\delta^{\ast}=p$}
        \psfrag{F}[][][0.9]{$\delta^{\ast}=p$}
        \psfrag{G}[][][1]{$w=0 / ?$}
        \psfrag{H}[][][1]{$w=0$}
        \psfrag{I}[][][1]{$w=1$}
        \psfrag{J}[][][1]{$w=?$}
        \centerline{\includegraphics[height=9.5cm]{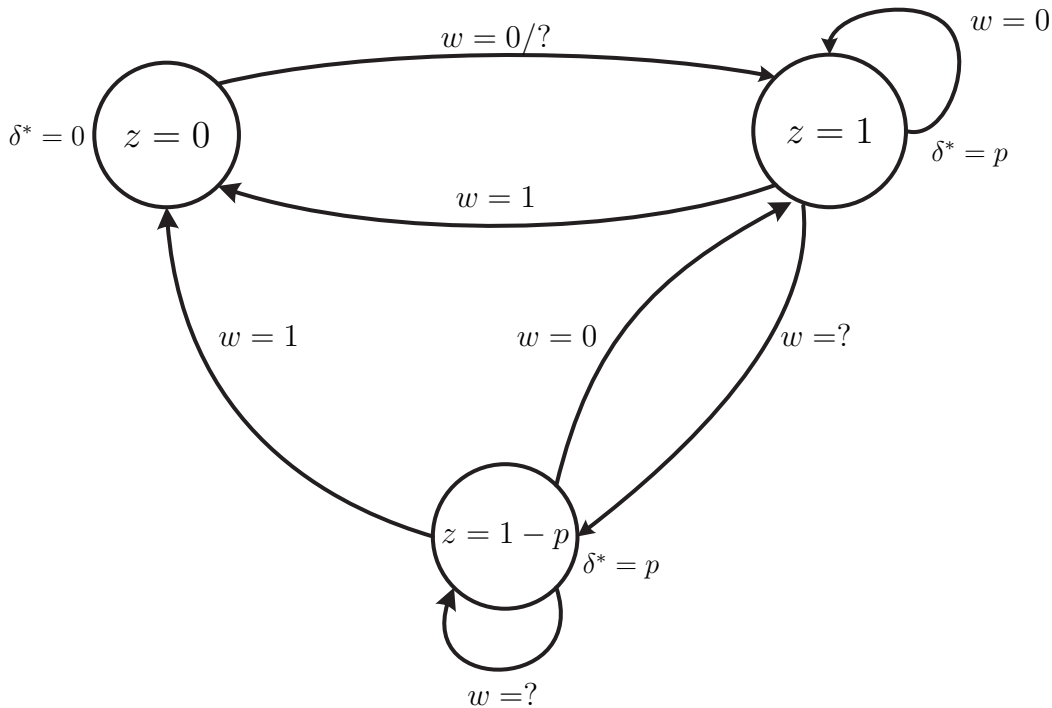}}
\caption{State diagram of the DP for the input-constrained BEC under an optimal policy.}
\label{fig:DP_scheme}
\end{figure}

The histogram for $\epsilon=0.5$, in Fig. \ref{fig:hist}, shows that under an optimal policy, $\delta^{\ast}$, the system evolves between three steady states. Moreover, the solution of the Bellman equation indicates that there exists an optimal stationary policy, and therefore, we look at the stationary phase of the DP. The states, $z$, take values in the finite set $\{0,1-p,1\}$, with $p\triangleq p_{\epsilon}$ (Eq. \eqref{eq:rho_z_defintion}); the subscript $\epsilon$ is omitted for convenience, but all details are discussed for a fixed $\epsilon\in[0,1]$ and its corresponding $p_{\epsilon}$. For each state, the optimal policy, $\delta^{\ast}$, is known from the Bellman equation and arrows can be drawn between the states as a function of the disturbance. The state diagram for our DP is presented in Fig. \ref{fig:DP_scheme}.

Converting the state diagram in Fig. \ref{fig:DP_scheme} into channel coding terms, using the formulation described in Table \ref{table:formulation}, results in an encoding procedure as described in Fig. \ref{fig:encoding_procedure}. Specifically, the states, $p(x_{t-1}=0|y^{t-1})$, take values from $\{0,1-p,1\}$. Each state has its corresponding action, $p(x_{t}=1|x_{t-1}=0)$, and the encoding procedure evolves as a function of the output $y_t$. Recall that $p(x_{t}=0|x_{t-1}=1)=1$, and therefore, the action $p(x_{t}=1|x_{t-1}=0)$ is sufficient to determine the transfer matrix between $X_{t-1}$ and $X_t$.
\begin{figure}[h!]
\centering
        \psfrag{A}[b][][.8]{$\mathbf{p(x_{t-1}=0|y^{t-1})}$}
        \psfrag{B}[t][][.8]{$\mathbf{=0}$}
        \psfrag{C}[b][][.8]{$\mathbf{p(x_{t-1}=0|y^{t-1})}$}
        \psfrag{D}[t][][.8]{$\mathbf{=1}$}
        \psfrag{E}[b][][.8]{$\mathbf{p(x_{t-1}=0|y^{t-1})}$}
        \psfrag{F}[t][][.8]{$\mathbf{=1-p}$}
        \psfrag{K}[b][][1]{$p(x_{t}=1|x_{t-1}=0)=0$}
        \psfrag{L}[b][][1]{$p(x_{t}=1|x_{t-1}=0)=p$}
        \psfrag{M}[r][][1]{$p(x_{t}=1|x_{t-1}=0)=\frac{p}{1-p}$}
        \psfrag{G}[][][1]{$y_t=0 / ?$}
        \psfrag{H}[][][1]{$y_t=0$}
        \psfrag{I}[][][1]{$y_t=1$}
        \psfrag{J}[][][1]{$y_t=?$}
        \psfrag{Z}[][][.7]{}
        \psfrag{X}[][][.7]{}
        \centerline{\includegraphics[height=10.5cm]{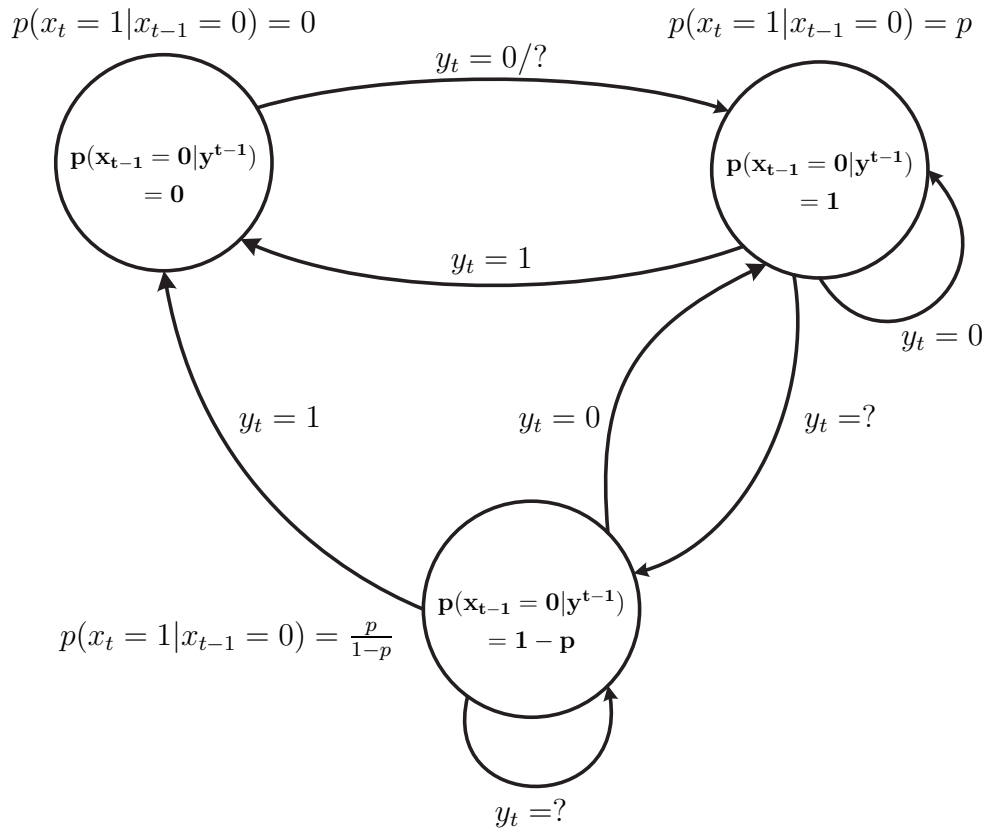}}
\caption{Optimal encoding procedure for the input-constrained BEC. This encoding procedure was achieved from Fig. \ref{fig:DP_scheme} by converting states, actions and disturbances into their corresponding channel coding terms.}
\label{fig:encoding_procedure}
\end{figure}

Let us explain how the encoding procedure evolves. We refer to the state $p(x_{t-1}=0|y^{t-1})=1$ as the \emph{ground state}, since this indicates that $'0'$ was received at the decoder and, therefore, the encoder is allowed to transmit any input to the channel. For the ground state, the next transmitted bit is distributed according to $\text{Ber}(p)$ and it is shown to be the optimal action.

Upon receiving $y_t=0$ at the decoder, the system remains at the ground state and the encoding procedure starts over again. When the output is $y_t=1$, the system moves to the state $p(x_{t-1}=0|y^{t-1})=0$. At this state, since the last input was necessarily $'1'$, the encoder is forced to transmit $'0'$. Therefore, the decoder knows that $'0'$ is the only legitimate input, and the system returns to the ground state regardless of whether the input was erased or not.

The remaining scenario to examine begins at the ground state and is followed by $y_t=?$. The optimal action at the lower state, $p(x_{t-1}=0|y^{t-1})=1-p$, suggests that if $'0'$ is erased, the new transmitted bit should be distributed according to $\text{Ber}(\frac{p}{1-p})$. The term $\frac{p}{1-p}$ is in the unit interval, since $p\leq \frac{1}{2}$. Additionally, the input constraint implies that if $'1'$ was erased then $'0'$ should be transmitted. Upon consecutive erasures, the encoder continues to transmit bits according to this policy. When an output is not an erasure, the system returns to the ground state, and this might take one or two time instances, depending on whether the (unerased) output bit is $'0'$ or $'1'$.

The main challenge is to understand how this encoding procedure can be interpreted as transmitting a message by the encoder. Let the messages be points in the unit interval, i.e., messages take values in the set $\mathcal{M}\triangleq\{\frac{k}{2^{nR}}\}_{k=0}^{2^{nR}-1}$. At each time instance, the unit interval contains sub-intervals with labels that can be $'0'$ or $'1'$, and the input to the channel is simply the label of the sub-interval containing the message. Such an association of messages into a specified interval has been done before in \cite{horstein_original,Kailath_scheme_one,Kim06_MA,shayevitz_posterior_mathcing}.

The partition into sub-intervals will be according to parameters $p$ and $q\triangleq\frac{p}{1-p}$, as described in Fig. \ref{fig:encoding_procedure}. When performing a partition at the ground state, the lower interval is labelled $'0'$ while the upper interval is labelled $'1'$. Before providing the precise encoding algorithm, it will be convenient to understand the labelling process in the example described in Fig. \ref{fig:encoding_example}.
\begin{figure}[t]
\centering
        \psfrag{A}[t][][.9]{$\bar{p}$}
        \psfrag{B}[t][][.9]{$\bar{p}\bar{q}$}
        \psfrag{C}[t][][.9]{$\bar{p}\bar{q}\bar{q}$}
        \psfrag{D}[t][][0.9]{$1$}
        \psfrag{Z}[t][][0.9]{$0$}
        \psfrag{E}[t][][0.9]{$\bar{p} + p\bar{q}$}
        \psfrag{L}[l][][0.6]{Messages}
        \psfrag{G}[l][][0.6]{Partition}
        \psfrag{F}[l][][.6]{Transmitted message}
        \psfrag{H}[r][][.7]{$y_1=?$}
        \psfrag{I}[r][][.7]{$y_2=?$}
        \psfrag{J}[l][][.7]{$y_2=0/1$}
        \psfrag{K}[l][][.7]{$y_3=0/1$}
        \psfrag{x}[l][][0.7]{\underline{\textbf{Legend:}}}
        \centerline{\includegraphics[height=8cm]{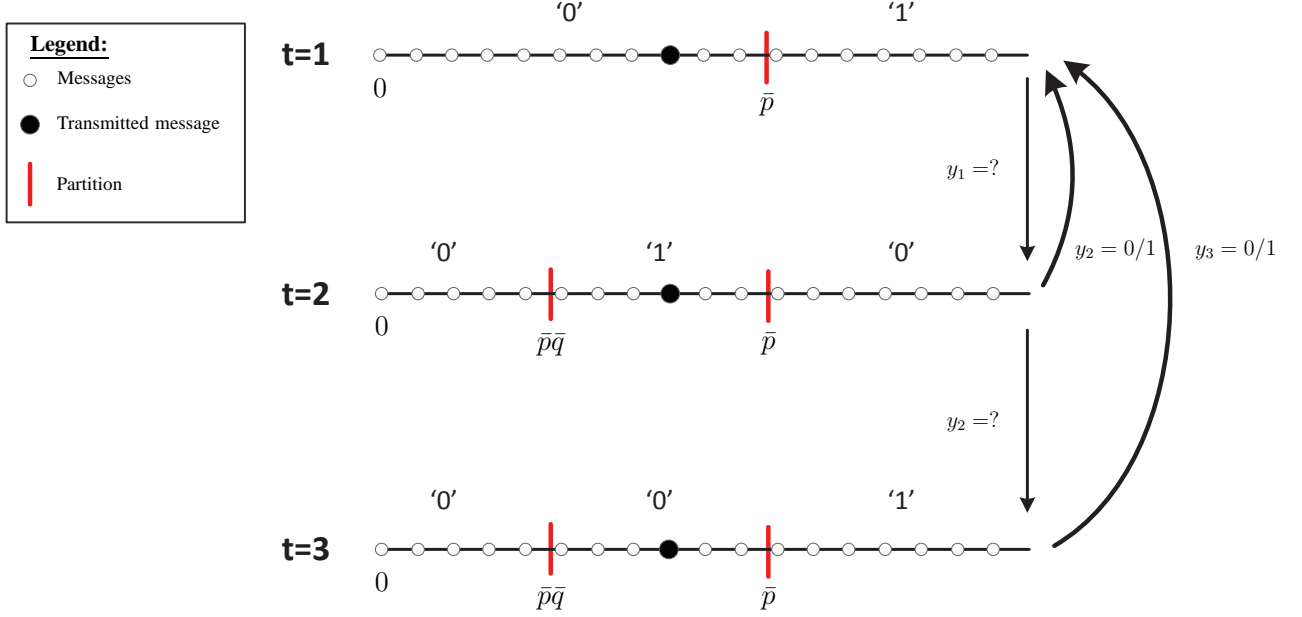}}
\caption{Example for transmitting the black-dot message using the encoding procedure in Fig. \ref{fig:encoding_procedure} for $3$ time instances. The initial partition at the ground state is according to $p$, and the encoder transmits $'0'$ since the black-dot message falls within $[0,\bar{p})$. Upon a successful transmission, the encoder moves back to the ground state and a new procedure begins. In case of erasure, we move to $t=2$, and the interval that was labelled $'0'$ is partitioned according to $q=\frac{p}{1-p}$. The input constraint is preserved since the interval $[\bar{p},1)$, that was labelled $'1'$, is now flipped to $'0'$. The encoder transmits $'1'$ since the message falls within $[\bar{p}\bar{q},\bar{p})$. In case of another erasure, a partition of $q$ should be performed for the intervals that are labelled $'0'$. These intervals are $[0,\bar{p}\bar{q})$ and $[\bar{p},1)$, which are sum up to $1-p$. Since $q=\frac{p}{1-p}$, we simply change the label of $[\bar{p},1)$ (which has length of $p$) to $'1'$, and the label of $[0,\bar{p}\bar{q})$ remains $'0'$. The input-constraint is preserved since $[\bar{p}\bar{q},\bar{p})$ is re-labelled as $'0'$. Upon another erasures, the labelling will be exchanged between the ones presented in $t=2$ and $t=3$ until a successful transmission. Note that the labelling at $t=1$ and $t=3$ are essentially the same.}
\label{fig:encoding_example}
\vspace{-6mm}
\end{figure}

As can be seen in Fig. \ref{fig:encoding_example}, all the proposed partitions in Fig. \ref{fig:encoding_procedure} can be encapsulated into two possible labellings. We denote the labelling at $t=1$ as $L_1$, and the labelling at $t=2$ as $L_2$. The initial labelling at the ground state is chosen as $L_1$, and upon erasure, the current labelling will be replaced with the other labelling. Note that changing the labelling $L_i$ with $L_j$ for $i\neq j$ preserves the input constraint and can be done simply by exchanging the labels of $[\bar{p}\bar{q},\bar{p})$ and $[\bar{p},1)$, while the label of $[0,\bar{p}\bar{q})$ remains $'0'$.

To summarize at this point, at each time instant, we have two possible labellings (which depend on the value of $\epsilon$) of the unit interval which define uniquely the mapping from messages to the channel input. The current labelling is determined only by the output tuple, $y^{t-1}$, and therefore, the decoder and encoder both agree on the latter.
\subsection{Capacity-achieving Coding Scheme}\label{subsec:coding_scheme}
At time instance $t-1$, the \textit{set of possible messages} is defined as $\mathcal{M}_{t-1}=\{m\in\mathcal{M}:p(m|y^{t-1})>0\}$, with $\mathcal{M}_{0}=\mathcal{M}$. The conditional distribution $p(m|y^{t-1})$ is calculated using Bayes' rule, using the fact that the encoding procedure and both labellings are revealed to all parties before transmission begins. Note that the set of possible messages can also be calculated at the encoder, since the output tuple, $y^{t-1}$, is available from the feedback.

Any received symbol at the decoder might reduce the set of potential messages, and a \textit{successful transmission} is defined as a transmission where the size of the set of possible messages is changed, namely, $|\mathcal{M}_{t}|<|\mathcal{M}_{t-1}|$. Specifically, a successful transmission can occur in one of two scenarios; the first is $y_{t}=1$, and the second is where $y_{t}=0$ and $y_{t-1}\neq 1$. Upon a successful transmission, the set of possible messages is calculated and expanded uniformly to the unit interval. To be precise, the messages in the set $\mathcal{M}_t$ take values in $\{\frac{k}{|\mathcal{M}_t|}\}^{|\mathcal{M}_t|-1}_{k=0}$. This transmission procedure continues repeatedly until the set of possible messages contains one message. The detailed encoding and decoding procedures are described in Algorithms \ref{alg:encoding} and \ref{alg:decoding}.

\begin{algorithm}[h!]
\caption{Encoding Procedure}
\label{alg:encoding}
\begin{algorithmic}
\While {Set of possible messages contains more than one message}		
   	\State Label the unit interval according to $L_1$.
    \State Transmit the label of the sub-interval containing the message.
    \While {Received symbol is an erasure}
        \State Exchange the labels of $[\bar{p}\bar{q},\bar{p})$ and $[\bar{p},1)$.
        \State Transmit the label of the sub-interval containing the message.
    \EndWhile
    \If {Received symbol is $'0'$}
        \State Denote the messages within sub-intervals which are labelled $'0'$ as the set of possible messages.
    \Else
        \State Denote the messages within sub-intervals which are labelled $'1'$ as the set of possible messages
        \State Transmit $'0'$.
    \EndIf
    \State Expand the set of possible messages to the unit interval.
\EndWhile
\end{algorithmic}
\end{algorithm}

\begin{algorithm}
\caption{Decoding Procedure}          
\label{alg:decoding}
\begin{algorithmic}                    
\While {Set of possible messages contains more than one message}
   	\State Label the unit interval according to $L_1$.
    \While {Received symbol is an erasure}
        \State Exchange the labels of $[\bar{p}\bar{q},\bar{p})$ and $[\bar{p},1)$.
    \EndWhile
    \If {Received symbol is $'0'$}
        \State Denote the messages within sub-intervals which are labelled $'0'$ as the set of possible messages.
    \Else
        \State Denote the messages within sub-intervals which are labelled $'1'$ as the set of possible messages.
        \State Ignore the next received symbol.
    \EndIf
    \State Expand the set of possible messages to the unit interval.
\EndWhile
\end{algorithmic}
\end{algorithm}

\textbf{Rate Analysis:}
The main feature of this coding scheme is that the length of the sub-interval that is labelled by $'1'$ is $p$. This property is recorded as Lemma \ref{lemma:constant_amount}.
\begin{lemma}\label{lemma:constant_amount}
At any step of the message transmission process, the lengths of the sub-intervals that are labelled by $'1'$ sum up to $p$.
\end{lemma}
\begin{proof}
Throughout transmission, there are two possible labellings; for $L_1$, the interval $[\bar{p},1)$ that is labelled $'1'$ has length of $p$, while for $L_2$, the interval $[\bar{p}\bar{q},\bar{p})$ has length of $\bar{p}q=p$.
\end{proof}
From Lemma \ref{lemma:constant_amount}, we note that the encoder transmits $'1'$ if message falls within sub-interval that has length of $p$. However, the messages are discrete points and a partition might fall between two messages. This implies that the transmitted bit is distributed as $\text{Ber}(p+e_i)$, where $e_i$ is a correction factor. In Appendix \ref{app:analysis_accurate}, it is shown that the correction factor has a negligible effect on the rate of the coding scheme. To simplify the derivations here, with some loss of accuracy, we say that each transmitted bit is distributed according to $\text{Ber}(p)$.

In the next lemma, we show that each successful transmission reduces the expected number of bits that is required to describe the set of possible messages by $H_b(p)$.
\begin{lemma}\label{lemma:const_amount}
With each successful transmission, the expected number of bits that describe the set of possible messages is reduced by $H_b(p)$.
\end{lemma}
\begin{proof}
Assume that the set of possible messages is of size $k$; upon a successful transmission, if $'0'$ is received then the new set of possible messages has size $\bar{p}{k}$, and if $'1'$ is received then its new size is $p{k}$. The expected number of bits that is required to describe the new set of possible messages is $\bar{p}\log_2(\bar{p}{k}) + p \log_2( p{k})=\log_2 k-H_b(p)$.
\end{proof}
The next step is to calculate the expected number of channel uses for a \emph{complete procedure}. We define a complete procedure to consist of all transmissions by the encoder starting at some time $t$ at which it is in the ground state, and ending at the first time $t' > t$ at which it returns to the ground state. In other words, a procedure is completed when a $'0'$ or $'1'$ is received at the decoder, including one extra channel use in the case when a $'1'$ has been received and has to be followed by $'0'$.

Let $N$ be a random variable corresponding to the number of channel uses within a complete procedure. The expected value of $N$ will be calculated by the law of total expectation. Define an indicator function \begin{align*}
\theta = \begin{cases} 0 &\mbox{if the received bit is $'0'$} \\
1 & \mbox{if the received bit is $'1'$}, \end{cases} \end{align*}
and consider,
\begin{align*}
  \mathbb{E}[N] &\stackrel{(a)}= \mathbb{E}[\mathbb{E}[N|\theta]] \\
   &\stackrel{(b)}= \mathbb{E}[\frac{1}{1-\epsilon} + \theta] \\
   &\stackrel{(c)}= \frac{1}{1-\epsilon} + p,
\end{align*}
where $(a)$ follows from the law of total expectation, $(b)$ follows from the fact that channel is memoryless and, therefore, $\frac{1}{1-\epsilon}$ is the expected value of time to receive a symbol which is not an erasure, and $(c)$ follows from $\mathbb{E}[\theta]=\Pr(\theta=1)$.

Finally, we prove the second part of Theorem \ref{theorem:main}, specifically, the rate of this coding scheme can be arbitrary close to the capacity expression, $C^{\text{fb}}_{\epsilon}$.
\begin{proof}
It follows from the law of large numbers that the rate of our coding scheme can be arbitrarily close to the expected number of received bits within a complete procedure divided by the expected number of channel uses within a complete procedure. In Lemma \ref{lemma:const_amount}, we showed that within a successful transmission, the expected number of received bits is $H_b(p)$. Moreover, the expected number of channel uses within a complete procedure is $\mathbb{E}[N]=\frac{1}{1-\epsilon}+p$. Therefore, the rate of the code can be arbitrarily close to $R=\frac{H_b(p)}{p + \frac{1}{1-\epsilon}}$.
\end{proof}
The above proof and Theorem \ref{theorem:sol_bellman} conclude the proof of our main result Theorem \ref{theorem:main}.
\section{Non-causal Capacity}\label{sec:upperbound}
In this section, we prove Theorem \ref{theorem:non_causal} by showing that $C_{\epsilon}^{\text{nc}}=\max_{0\leq p \leq \frac{1}{2}} \frac{H_{b}(p)}{p+{\frac{1}{1-\epsilon}}}$. Operational considerations of non-causal and feedback capacities reveal the trivial inequality $C^{\text{nc}}_{\epsilon}\geq C^{\text{fb}}_{\epsilon}$. Furthermore, we derive in this section an upper-bound on $C^{\text{nc}}_{\epsilon}$, which is equal to $C_{\epsilon}^{\text{fb}}$, and this concludes the proof of Theorem \ref{theorem:non_causal} with $C^{\text{nc}}_{\epsilon}= C^{\text{fb}}_{\epsilon}$.

The next lemma shows that it is sufficient to consider encoders which transmit $'0'$ if erasure occurs, i.e., $x_{i}=0$ if $\theta_{i}=1$. The intuition behind this lemma is that replacing erased ones with zeros does not effect the output sequence, while the input-constraint is not violated.
\begin{lemma}\label{lemma:strict_encoder}
For any $(1,2^{nR},(1,\infty))$ code $\mathcal{C}$ with probability of error $P_{e}^{(n)}$, there exists a $(n,2^{nR},(1,\infty))$ code $\mathcal{C}'$ with probability of error $P_{e}^{(n)}$, satisfying
\begin{equation*}
f_i(m,y^{i-1},\theta_i=1)=0, i=1,\dots,n \ \forall (m,y^{i-1}).
\end{equation*}
\end{lemma}
\begin{proof}
For any $(1,2^{nR},(1,\infty))$ code $\mathcal{C}$ consisting of encoding functions, $\{f_i(\cdot)\}_{i=1}^{n}$, and a decoding function $\Psi(\cdot)$ with probability of error $P_{e}^{(n)}$, define a new sequence of encoding functions as follows:
\begin{equation}\label{eq:new_encoding}
    \tilde{f}_{i}(m,y^{i-1},\theta_i)=\left\{\begin{array}{cc}
    f_{i}(m,y^{i-1},\theta_i) & \text{if } \theta_i=0, \\
    0 & \text{if } \theta_i=1, \end{array}\right.
\end{equation}
 for all $(m,y^{i-1})$ and $i=1,\dots,n$.
We argue that $\{\tilde{f}_i(\cdot)\}_{i=1}^{n}$ and the original decoding function $\Psi(\cdot)$ determine a new code with the same probability of error $P_{e}^{(n)}$.
First, the set of encoding functions, $\{\tilde{f}_i(\cdot)\}_{i=1}^{n}$, satisfies the input constraint, since we replaced ones with zeros. Further, the output sequence is not affected by our modification, since we replaced only bits that are erased, and therefore, our new code also has probability of error $P_{e}^{(n)}$.
\end{proof}

We introduce $(1,\infty,\text{Ber}(\epsilon))$-RLL \textit{encoder}, which outputs sequences $X^n$ that satisfies two constraints:
\begin{enumerate}
  \item The $(1,\infty)$-RLL constraint.
  \item $X_i=0$ if $\theta_i=1$ (the constraint induced by Lemma \ref{lemma:strict_encoder}).
\end{enumerate}
The second constraint can be viewed as a "random constraint" since $\theta_i\sim\text{Ber}(\epsilon)$, while the first constraint is a deterministic constraint. Thus, the $(1,\infty,\text{Ber}(\epsilon))$-RLL encoder combines both deterministic and random constraints.

The entropy rate of $(1,\infty,\text{Ber}(\epsilon))$-RLL encoder is measured by $\lim_{n\rightarrow\infty} \sum_{i=1}^n H(X_i|X^{i-1},\theta^i)$ since this is the available information at the encoder. The next lemma provides an upper bound on the entropy rate of sequences that can be generated by a $(1,\infty,\text{Ber}(\epsilon))$-RLL encoder.
\begin{lemma}\label{lemma:entropy_rate}
The entropy rate of sequences that are generated by a $(1,\infty,\text{Ber}(\epsilon))$-RLL encoder is upper bounded by
$\max_{0\leq p\leq \frac{1}{2}}\frac{H_{b}(p)}{p+\frac{1}{1-\epsilon}}$.
\end{lemma}
\begin{proof}
Recall that the encoder can choose its output bit, $x_i$, only if $x_{i-1}=\theta_i=0$; we parameterize this by $p(x_i=1|x_{i-1}=0,\theta_i=0)=p$, where $p\in[0,1]$. Now, consider the transition probability matrix of the chain $X^n$,
\begin{align*}
    \mathbf{Q} =& \left[
                   \begin{array}{cc}
                     \epsilon + \bar{\epsilon}\bar{p} & \bar{\epsilon}p \\
                     1 & 0 \\
                   \end{array}
                 \right],
\end{align*}
where the transition probability $\epsilon + \bar{\epsilon}\bar{p}$ was calculated by
\begin{equation*}
p(x_i=0|x_{i-1}=0)=\sum_{\theta_i}p(x_i=0,\theta_i|x_{i-1}=0).
\end{equation*}
The stationary distribution of this chain is $[x^{\ast}(0) \ x^{\ast}(1)] = [\frac{1}{1+\bar{\epsilon}p} \ \frac{\bar{\epsilon}p}{1+\bar{\epsilon}p}]$.

Consider the next upper bound for some $i$,
\begin{align}\label{eq:lemma_non_proof}
    H(X_i|X^{i-1},\theta^i)&\stackrel{(a)}\leq H(X_i|X_{i-1},\theta_i) \nonumber\\
     &\stackrel{(b)}= H(X_i|X_{i-1},\theta_i=0)\bar{\epsilon}\nonumber\\
     &\stackrel{(c)}= H(X_i|x_{i-1}=0,\theta_i=0)p(x_{i-1}=0|\theta_i=0)\bar{\epsilon} \nonumber\\
     &\stackrel{(d)}= H_b(p)p(x_{i-1}=0)\bar{\epsilon}
\end{align}
where $(a)$ follows conditioning reduces entropy, $(b)$ follows from $H(X_i|X_{i-1},\theta_i=1)=0$, $(c)$ follows from $H(X_i|x_{i-1}=1,\theta_i=0)=0$, and $(d)$ follows from the fact that $X_{i-1}$ is independent of $\theta_i$ and substituting the parameter $p$.

By substituting the stationary distribution $p(x_{i-1}=0)=x^{\ast}(0)$ into \eqref{eq:lemma_non_proof}, we see that the entropy rate of the chain is upper bounded by $\frac{\bar{\epsilon}H_b(p)}{1+\bar{\epsilon}p}$, for some $p\in[0,1]$. This term can also be written as $\frac{H_b(p)}{\frac{1}{\bar{\epsilon}}+p}$, and the parameter $p$ need be maximized only on $[0,0.5]$ from Lemma \ref{lemma:concave}.
\end{proof}
The rate of the message $M$ is upper bounded by the entropy rate of sequences that can be generated by a $(1,\infty,\text{Ber}(\epsilon))$-RLL encoder, and this concludes the proof of Theorem \ref{theorem:non_causal} with
\begin{align*}
  C_{\epsilon}^{\text{nc}} &\leq  \max_{0\leq p\leq \frac{1}{2}}\frac{H_{b}(p)}{p+\frac{1}{1-\epsilon}}\\
   &=  C_{\epsilon}^{\text{fb}}.
\end{align*}
\section{Conclusions}\label{sec:conclusions}
We considered the setup of an input-constrained erasure channel with feedback and found its capacity using equivalent DP. We then pursued the complementary derivation of a simple and error-free capacity-achieving coding scheme, which we found using the strong relation between optimal policies in DP and encoding procedures in channel coding. Moreover, we have shown that the capacity remains the same even if the erasure is known non-causally to the encoder.

Following the theorem that feedback does not increase the capacity of a memoryless channel \cite{shannon56}, Shannon also argued that this theorem can be extended to channels with memory if the channel state can be computed at the encoder. Our system setting falls into this criteria, since the previous input of the channel can be thought of as the channel state. The proof for Shannon's argument was omitted, although not trivial, and still stands as a conjecture.

Following Shannon's conjecture, it could be interesting to derive the capacity of the input-constrained erasure channel with delayed feedback, namely, when the input to the channel at time $i$ depends on the message and the tuple $Y^{i-\nu}$, where $\nu$ is the delay of the feedback. Dynamic programming formulation for the delayed-feedback capacity is feasible and could shed light on Shannon's conjecture and on the capacity of the input-constrained erasure channel without feedback. Furthermore, a model with arbitrary delayed feedback will provide a new upper bound for the capacity of the input-constrained BEC without feedback, a problem that is wide open.
\appendices
\section{Proof of Lemma }\label{app:lemma concave}
\begin{proof}[Proof of Lemma \ref{lemma:concave}]
\begin{itemize}
\item A sufficient condition for the concavity of a function $f(p)$ is that the second derivative is negative for any value of $p$. We denote $k=\frac{1}{1-\epsilon}$ and find a condition on $k$ such that the second derivative is negative. To simplify the derivations, we take $H_{b}(\cdot)$ to be the binary entropy with the natural logarithm base, since multiplication by a constant does not effect concavity. Calculation shows that
\begin{align}
  \frac{d^2}{dp^2}\left(\frac{H_b(p)}{p+k}\right) &= \frac{\frac{(p+k)^2}{p(p-1)} -2k\ln\left(\frac{1-p}{p}\right) - 2 \ln (1-p)}{p^3}.
\end{align}
It suffices to examine the sign of the numerator, since $p^3\geq 0$. Define $g(p)\triangleq \frac{(p+k)^2}{p(p-1)} -2k\ln\left(\frac{1-p}{p}\right) - 2 \ln (1-p)$. Derivation of the maximum for $g(p)$ shows that it has only one maximum, which is at $p=\frac{1}{2}$. Substituting $g(\frac{1}{2})= -4(\frac{1}{2}+k)^2 + 2 \ln 2$. It then follows that $g(p)\leq 0, \forall p\in[0,1]$ if and only if $k\geq \sqrt{\frac{1}{2}\ln 2}- \frac{1}{2} \sim 0.088$.
\item Derivation of the first derivative of $f(p)$ shows that the derivative is equal to zero if and only if $p^{\frac{1}{\bar{\epsilon}}}=(1-p)^{1+\frac{1}{\bar{\epsilon}}}$ holds. The uniqueness of the maximum point follows from the fact that $p^{\frac{1}{\bar{\epsilon}}}$ increases as $p$ grows, while $(1-p)^{1+\frac{1}{\bar{\epsilon}}}$ decreases with a growing $p$.

 Now, assume that the maximum is $p_m\in(\frac{1}{2},1]$. Symmetry of the binary entropy function implies $H_b(p_m) = H_b(\bar{p}_m)$, and therefore, it is sufficient to examine the denominator. Since both arguments $p_m,\bar{p}_m\in[0,1]$, it then follows that $f(p_m)<f(\bar{p}_m)$, which is a contradiction.
\item This property follows from substituting the relation $p^{\frac{1}{\bar{\epsilon}}}=(1-p)^{1+\frac{1}{\bar{\epsilon}}}$ into the function $f(p)$.

\end{itemize}
\end{proof}
\section{Proof of Theorem \ref{theorem:sol_bellman}}\label{app:proof_bellman}
The next lemma is technical and will be useful in the proof of Theorem \ref{theorem:sol_bellman}.
\begin{lemma}\label{lemma:same_max}
    The function $f_{\epsilon}(z)= \bar{\epsilon}H_b(z) -z\bar{\epsilon}\frac{H_b(p_\epsilon)}{p_\epsilon+\frac{1}{\bar{\epsilon}}}$ is concave on $[0,1]$ and its maximum is at $z=p_\epsilon$, where $p_\epsilon = \argmax_{0\leq p \leq \frac{1}{2}} \frac{H_b(p)}{p+\frac{1}{1-\epsilon}}$.
\end{lemma}
\begin{proof}[Proof of Lemma \ref{lemma:same_max}]
The concavity of $f_{\epsilon}(z)$ on $z\in[0,1]$ follows from the concavity of the binary entropy function, and therefore, it suffices to show that the first derivative of $f_{\epsilon}(z)$ at $p_{\epsilon}$ is equal to zero.
The definition of $p_{\epsilon}$, \eqref{eq:rho_z_defintion}, and Lemma \ref{lemma:concave} imply the relation, $\frac{d}{dz}\left[ \frac{H_b(z)}{z+\frac{1}{\bar{\epsilon}}}\right]\vline_{z=p_\epsilon}=0$, which is equivalent to
\begin{align}\label{relation}
        H^{'}_{b}(p_\epsilon)(p_\epsilon+\frac{1}{\bar{\epsilon}}) - H_{b}(p_\epsilon)=0.
\end{align}

The first derivative of $f_{\epsilon}(z)$ at the point $p_\epsilon$ is:
\begin{align*}
  \frac{d}{dz}\left[\bar{\epsilon}H_b(z) - z\bar{\epsilon}\frac{H_b(p_\epsilon)}{p_\epsilon+\frac{1}{\bar{\epsilon}}}\right]\vline_{z=p_\epsilon}
        &= \left(\bar{\epsilon}H^{'}_b(z) - \bar{\epsilon}\frac{H_b(p_\epsilon)}{p_\epsilon+\frac{1}{\bar{\epsilon}}}\right)\vline_{z=p_\epsilon} \\
        &= \frac{\bar{\epsilon}H^{'}_b(z)(p_\epsilon+\frac{1}{\bar{\epsilon}}) - \bar{\epsilon}H_{b}(p_\epsilon)}{p_\epsilon+\frac{1}{\bar{\epsilon}}} \\
        &\stackrel{(a)}= 0.
\end{align*}
where $(a)$ follows from \eqref{relation}.
\end{proof}
We proceed to the proof of Theorem \ref{theorem:sol_bellman}.
\begin{proof}[Proof of Theorem \ref{theorem:sol_bellman}]
Substituting $z=0$ into \eqref{eq_our_bellman} yields $\rho_{\epsilon} + h_{\epsilon}(0) = h_{\epsilon}(1)$. It can be shown that if $h_{\epsilon}(z)$ solves \eqref{eq_our_bellman}, then any function of the form $h_{\epsilon}(z) + constant $ also solves this equation. Therefore, we can fix $h_{\epsilon}(0)=0$, which implies that $h_{\epsilon}(1) = \rho_{\epsilon}$. It then follows that the DP operator with the function $h_{\epsilon}^{\ast}(z)$ is:
\begin{align*}
(Th_{\epsilon}^{\ast})(z)&=\sup_{0\leq\delta \leq z} \bar{\epsilon}H_{b}(\delta) + \bar{\epsilon}(1-\delta)\frac{H_b(p_\epsilon)}{p_\epsilon+\frac{1}{\bar{\epsilon}}} +\epsilon h_{\epsilon}^{\ast}(1-\delta).
\end{align*}
Now, the term $h_{\epsilon}^{\ast}(1-\delta)$ is calculated for two cases:
\begin{align}\label{eq:value_cases}
     h_{\epsilon}^{\ast}(1-\delta) =
     \begin{cases} \bar{\epsilon}H_b(\delta) - (1-\delta)\bar{\epsilon}\frac{H_b(p_\epsilon)}{p_\epsilon+\frac{1}{\bar{\epsilon}}} &\mbox{if } 1- \delta \leq p_\epsilon  \\
        \frac{H_b(p_\epsilon)}{p_\epsilon+\frac{1}{\bar{\epsilon}}} & \mbox{if } 1-\delta \geq p_\epsilon.
     \end{cases}
\end{align}
To complete the proof, we have three cases for calculating the operator $(Th_{\epsilon}^{\ast})(z)$:
\begin{itemize}
\item
For $0\leq z < p_\epsilon$, the constraint $0\leq\delta\leq z$ implies that $0\leq\delta <p_\epsilon$, and from \eqref{eq:value_cases}, we have $h^{\ast}(1-\delta) = \frac{H_b(p_\epsilon)}{p_\epsilon+\frac{1}{\bar{\epsilon}}}$.
         Let us show that \eqref{eq_our_bellman} is satisfied:
         \begin{align*}
            (Th_{\epsilon}^{\ast})(z) &= \sup_{0\leq\delta \leq z} \bar{\epsilon}H_{b}(\delta) + \bar{\epsilon}(1-\delta)\frac{H_b(p_\epsilon)}{p_\epsilon+\frac{1}{\bar{\epsilon}}} +
                              \epsilon \frac{H_b(p_\epsilon)}{p_\epsilon+\frac{1}{\bar{\epsilon}}}\\
                           &= \sup_{0\leq\delta \leq z} \bar{\epsilon}H_{b}(\delta) -\delta \bar{\epsilon}\frac{H_b(p_\epsilon)}{p_\epsilon+\frac{1}{\bar{\epsilon}}} +
                              \frac{H_b(p_\epsilon)}{p_\epsilon+\frac{1}{\bar{\epsilon}}} \\
                           &\stackrel{(a)}= \bar{\epsilon}H_{b}(z) -z \bar{\epsilon}\frac{H_b(p_\epsilon)}{p_\epsilon+\frac{1}{\bar{\epsilon}}} +
                              \frac{H_b(p_\epsilon)}{p_\epsilon+\frac{1}{\bar{\epsilon}}} \\
                           &= h_{\epsilon}^{\ast}(z) + \rho_{\epsilon}^{\ast},
        \end{align*}
        where $(a)$ follows from Lemma \ref{lemma:same_max}.
\item
        For $p_{\epsilon} \leq z < 1 - p_{\epsilon}$, the same calculation as for the previous interval shows that $h^{\ast}(1-\delta) = \frac{H_b(p_\epsilon)}{p_\epsilon+\frac{1}{\bar{\epsilon}}}$ for all $\delta\in[0,1 - p_{\epsilon}]$. Let us show that \eqref{eq_our_bellman} is satisfied:
        \begin{align*}
            (Th^{\ast})(z) &= \sup_{0\leq\delta \leq z} \bar{\epsilon}H_{b}(\delta) + \bar{\epsilon}(1-\delta)\frac{H_b(p_\epsilon)}{p_\epsilon+\frac{1}{\bar{\epsilon}}} +
                              \epsilon \frac{H_b(p_\epsilon)}{p_\epsilon+\frac{1}{\bar{\epsilon}}}\\
                           &= \sup_{0\leq\delta \leq z} \bar{\epsilon}H_{b}(\delta) -\delta \bar{\epsilon}\frac{H_b(p_\epsilon)}{p_\epsilon+\frac{1}{\bar{\epsilon}}} +
                              \frac{H_b(p_\epsilon)}{p_\epsilon+\frac{1}{\bar{\epsilon}}} \\
                           &\stackrel{(a)}= \bar{\epsilon}H_{b}(p_\epsilon) - p_\epsilon \bar{\epsilon}\frac{H_b(p_\epsilon)}{p_\epsilon+\frac{1}{\bar{\epsilon}}} +
                              \frac{H_b(p_\epsilon)}{p_\epsilon+\frac{1}{\bar{\epsilon}}} \\
                           &=  \frac{H_b(p_\epsilon)}{p_\epsilon+\frac{1}{\bar{\epsilon}}} + \frac{H_b(p_\epsilon)}{p_\epsilon+\frac{1}{\bar{\epsilon}}} \\
                           &= h_{\epsilon}^{\ast}(z) + \rho_{\epsilon}^{\ast},
        \end{align*}
        where $(a)$ follows from Lemma \ref{lemma:same_max}.
\item
        For $1 - p_{\epsilon} \leq z \leq 1$, the function $h_{\epsilon}^{\ast}(1-\delta)$ can have different terms, and therefore, we separate:
        \begin{align*}
            (Th^{\ast})(z) &= \max\bigg( \sup_{0\leq\delta \leq 1-p_{\epsilon}} \bar{\epsilon}H_{b}(\delta) +\bar{\epsilon}(1-\delta)\frac{H_b(p_\epsilon)}{p_\epsilon+\frac{1}{\bar{\epsilon}}} +
                              \epsilon \frac{H_b(p_\epsilon)}{p_\epsilon+\frac{1}{\bar{\epsilon}}},\\
                           &   \sup_{1 - p_\epsilon \leq \delta \leq z} \bar{\epsilon}H_{b}(\delta) + \bar{\epsilon}(1-\delta)\frac{H_b(p_\epsilon)}{p_\epsilon+\frac{1}{\bar{\epsilon}}} +  \epsilon[\bar{\epsilon}H_b(\delta) - (1-\delta)\bar{\epsilon}\frac{H_b(p_\epsilon)}{p_\epsilon+\frac{1}{\bar{\epsilon}}}]
                              \bigg)\\
                           &\stackrel{(a)}= \max \bigg(   \frac{H_b(p_\epsilon)}{p_\epsilon+\frac{1}{\bar{\epsilon}}} + \frac{H_b(p_\epsilon)}{p_\epsilon+\frac{1}{\bar{\epsilon}}}, \sup_{1 - p_\epsilon \leq \delta \leq z} \bar{\epsilon}(1+\epsilon)H_{b}(\delta) + \bar{\epsilon}\bar{\epsilon}(1-\delta)\frac{H_b(p_\epsilon)}{p_\epsilon+\frac{1}{\bar{\epsilon}}} \bigg)\\
                           &\stackrel{(b)}= 2\frac{H_b(p_\epsilon)}{p_\epsilon+\frac{1}{\bar{\epsilon}}} \\
                           &= h_{\epsilon}^{\ast}(z) + \rho_{\epsilon}^{\ast},
            \end{align*}
        where $(a)$ follows from Lemma \ref{lemma:same_max}, and $(b)$ follows from
        \begin{align*}
             \sup_{1 - p_\epsilon \leq \delta \leq z} \bar{\epsilon}(1+\epsilon)H_{b}(\delta) + \bar{\epsilon}\bar{\epsilon}(1-\delta)\frac{H_b(p_\epsilon)}{p_\epsilon+\frac{1}{\bar{\epsilon}}}
             & \leq \sup_{1 - p_\epsilon \leq \delta \leq z} \bar{\epsilon}(1+\epsilon)H_{b}(\delta) + \sup_{1 - p_\epsilon \leq \delta \leq z}
             \bar{\epsilon}\bar{\epsilon}(1-\delta)\frac{H_b(p_\epsilon)}{p_\epsilon+\frac{1}{\bar{\epsilon}}}\\
            &= \bar{\epsilon}(1+\epsilon)H_{b}(1 - p_\epsilon) + \bar{\epsilon}\bar{\epsilon}(1-(1 - p_\epsilon))\frac{H_b(p_\epsilon)}{p_\epsilon+\frac{1}{\bar{\epsilon}}}\\
            &= \frac{H_b(p_\epsilon)}{p_\epsilon+\frac{1}{\bar{\epsilon}}}\left[ 2\bar{\epsilon}p_{\epsilon}+1+\epsilon\right] \\
            &\leq 2\frac{H_b(p_\epsilon)}{p_\epsilon+\frac{1}{\bar{\epsilon}}}.
        \end{align*}
\end{itemize}
\end{proof}
\section{Accurate rate analysis}\label{app:analysis_accurate}
The rate analysis in Section \ref{sec:relation_coding} was simplified by assuming that each transmitted bit is $\text{Ber}(p)$. Here, we show precisely that our coding scheme can be arbitrary close to $C_{\epsilon}^{\text{fb}}$. The idea is to separate the coding scheme into two parts using a parameter $\lambda$, which is a fixed constant. First, we use the coding scheme from Section \ref{subsec:coding_scheme} to transmit a large number, $nR-\lambda$, of message bits, while a different coding scheme will be used to transmit the remaining $\lambda$ bits. We show that the rate of the overall scheme is essentially determined by the rate of the first coding scheme.
The next lemma will be used for the rate analysis of the first coding scheme,
\begin{lemma}\label{lemma:correction}
Each transmitted bit, $X_i$, can be chosen to be distributed as $\text{Ber}(p-e_i)$, where $0\leq e_i<\frac{1}{|\mathcal{M}_{i-1}|}$.
\end{lemma}
\begin{proof}
Assume that at time $i$, a procedure begins and its corresponding set of possible messages is $\mathcal{M}_{i-1}$. According to $L_1$, the number of messages that are labelled $'1'$ is $\floor{p|\mathcal{M}_{i-1}|}$, where $\floor{\cdot}$ is the floor operator. The resulting input distribution is $X_i\sim\text{Ber}(\frac{\floor{p|\mathcal{M}_{i-1}|}}{|\mathcal{M}_{i-1}|})$, which can be written also as $X_i\sim\text{Ber}(p-e_i)$ since $p-\frac{1}{|\mathcal{M}_{i-1}|}< \frac{\floor{p|\mathcal{M}_{i-1}|}}{|\mathcal{M}_{i-1}|}\leq p$.

In case of erasure at time $i$, recall that the number of messages that were labelled $'0'$ in $L_1$ is greater than the number of messages labelled $'1'$, and thus, we are able to construct the labelling $L_2$ as follows; $\floor{p|\mathcal{M}_{i-1}|}$ messages that were labelled $'0'$ at the previous transmission are flipped to $'1'$, and all the remaining messages are labelled $'0'$. It is clear that the input distribution is preserved in this case, and upon consecutive erasures, $L_1$ and $L_2$ are being exchanged and the input distribution is not changed. Note that the choices of labelling are made in advance and both encoder and decoder agree on current labelling.
\end{proof}

The encoding procedure occurs repeatedly and is over when the set of possible messages is less or equal than $2^\lambda$. Denote by $e_1,e_2,\dots, e_k$ the correction factors for the $k$ successful transmissions until the scheme is over. Following the same derivations in Section \ref{sec:relation_coding}, it follows that the rate is $\tilde{R} = \frac{\sum_{i=1}^k H_b(p-e_i)}{k(\frac{1}{1-\epsilon}+p) - \sum_{i=1}^{k}e_i}$.

For the $\lambda$ remaining bits, we perform a code where a bit of message is followed by zero and this pair is transmitted repeatedly until a successful transmission. Thus, to send the message bit $'0'$, the pair $'00'$ is repeated until $'00'$ or $'0?'$ are received, and to send the message bit $'1'$, the bits $'10'$ are repeatedly transmitted until a $'1'$ is received. The decoding for this scheme is straightforward, and calculation of the rate gives that $\bar{R}=\frac{1-\epsilon}{2}$.

To summarize, the average rate for the overall coding scheme is
\begin{align*}
R&= \left(\frac{nR-\lambda}{nR}\right) \tilde{R} + \left(\frac{\lambda}{nR}\right)\bar{R}.
\end{align*}
Consider the next lower bound on $R$,
\begin{align*}\label{lower}
R&=\left(\frac{nR-\lambda}{nR}\right)\frac{\sum_{i=1}^k H_b(p-e_i)}{k(\frac{1}{1-\epsilon}+p) - \sum_{i=1}^{k}e_i} + \left(\frac{\lambda}{nR}\right)\frac{1-\epsilon}{2}\nonumber\\
&\geq \left(\frac{nR-\lambda}{nR}\right)\frac{k\min_i H_b(p-e_i)}{k(\frac{1}{1-\epsilon}+p) - k\min_i e_i} + \left(\frac{\lambda}{nR}\right)\frac{1-\epsilon}{2}\nonumber\\
&\stackrel{(a)}\geq \left(\frac{nR-\lambda}{nR}\right)\frac{H_b(p-2^{-\lambda})}{\frac{1}{1-\epsilon}+p} + \left(\frac{\lambda}{nR}\right)\frac{1-\epsilon}{2},
\end{align*}
where $(a)$ follows from Lemma \ref{lemma:correction}, namely, $e_i\in[0,2^{-\lambda})$ for $i=1,\dots,k$.

Letting $n\rightarrow\infty$, we see that $R^{\ast}= \frac{H_b(p-2^{-\lambda})}{\frac{1}{1-\epsilon}+p}$ is achievable. Thus, by choosing $\lambda$ to be arbitrarily large (but still finite), we can make $R^\ast$ arbitrarily close to the capacity $C_\epsilon^\text{fb}$.
\section*{Acknowledgment}
The authors would like to thank Yonglong Li and Guangyue Han for providing us the data for the lower bound in plotted in Fig. \ref{fig:comparison}.
\bibliography{ref}

\begin{thebibliography}{10}
\providecommand{\url}[1]{#1}
\csname url@samestyle\endcsname
\providecommand{\newblock}{\relax}
\providecommand{\bibinfo}[2]{#2}
\providecommand{\BIBentrySTDinterwordspacing}{\spaceskip=0pt\relax}
\providecommand{\BIBentryALTinterwordstretchfactor}{4}
\providecommand{\BIBentryALTinterwordspacing}{\spaceskip=\fontdimen2\font plus
\BIBentryALTinterwordstretchfactor\fontdimen3\font minus
  \fontdimen4\font\relax}
\providecommand{\BIBforeignlanguage}[2]{{%
\expandafter\ifx\csname l@#1\endcsname\relax
\typeout{** WARNING: IEEEtran.bst: No hyphenation pattern has been}%
\typeout{** loaded for the language `#1'. Using the pattern for}%
\typeout{** the default language instead.}%
\else
\language=\csname l@#1\endcsname
\fi
#2}}
\providecommand{\BIBdecl}{\relax}
\BIBdecl

\bibitem{Shannon48}
C.~E. Shannon, ``A mathematical theory of communication,'' \emph{Bell Syst.
  Tech. J.}, vol.~27, pp. 379--423 and 623--656, 1948.

\bibitem{Blahut72}
R.~E. Blahut, ``Computation of channel capacity and rate-distortion
  functions,'' \emph{IEEE Trans. Inf. Theory}, vol.~18, no.~4, pp. 460--473,
  Jul. 1972.

\bibitem{Arimoto72}
S.~Arimoto, ``An algorithm for computing the capacity of arbitrary discrete
  memoryless channels,'' \emph{Information Theory, IEEE Transactions on},
  vol.~18, no.~1, pp. 14--20, Jan. 1972.

\bibitem{vontobel_generalization}
P.~Vontobel, A.~Kavcic, D.~Arnold, and H.-A. Loeliger, ``A generalization of
  the {B}lahut-{A}rimoto algorithm to finite-state channels,'' \emph{IEEE
  Trans. Inf. Theory}, vol.~54, no.~5, pp. 1887--1918, May 2008.

\bibitem{han_constrained_BSC_BEC}
G.~Han and B.~Marcus, ``Asymptotics of entropy rate in special families of
  hidden {M}arkov chains,'' \emph{IEEE Trans. Inf. Theory}, vol.~56, no.~3, pp.
  1287--1295, Mar. 2010.

\bibitem{wolf_RLL}
E.~Zehavi and J.~Wolf, ``On runlength codes,'' \emph{IEEE Trans. Inf. Theory},
  vol.~34, no.~1, pp. 45--54, Jan. 1988.

\bibitem{han_RLL_BSC}
G.~Han and B.~Marcus, ``Concavity of the mutual information rate for
  input-restricted memoryless channels at high {SNR},'' \emph{IEEE Trans. Inf.
  Theory}, vol.~58, no.~3, pp. 1534--1548, Mar. 2012.

\bibitem{Marcus98}
B.~H. Marcus, R.~M. Roth, and P.~H. Siegel, ``Constrained systems and coding
  for recording channels,'' in \emph{Handbook of Coding Theory}, V.~S. Pless
  and W.~C. Huffman, Eds.\hskip 1em plus 0.5em minus 0.4em\relax Amsterdam,
  Netherlands: Elsevier, 1998, pp. 1635--1764.

\bibitem{Immink04}
K.~Immink, \emph{Codes for Mass Data Storage Systems}.\hskip 1em plus 0.5em
  minus 0.4em\relax Rotterdam, The Netherlands: Shannon Foundation, 2004.

\bibitem{osvalso_charge_battery}
A.~Fouladgar, O.~Simeone, and E.~Erkip, ``Constrained codes for joint energy
  and information transfer,'' \emph{IEEE Trans. Commun.}, vol.~62, no.~6, pp.
  2121--2131, Jun. 2014.

\bibitem{shannon56}
C.~Shannon, ``The zero error capacity of a noisy channel,'' \emph{IEEE Trans.
  Inf. Theory}, vol.~2, no.~3, pp. 8--19, Sep. 1956.

\bibitem{constrained_erasure_nofeedback_isit}
Y.~Li and G.~Han, ``Input-constrained erasure channels: Mutual information and
  capacity,'' in \emph{Proc. IEEE Int. Symp. Inf. Theory (ISIT 2014)},
  Honolulu, Hawaii, USA, June 29 - July 4, 2014.

\bibitem{Tatikonda00}
S.~C. Tatikonda, ``Control under communication constraints,'' Ph.{D}.
  dissertation, Massachusetts Institute of Technology, Cambridge, MA, 2000.

\bibitem{Yang05}
S.~Yang, A.~Kav\u{c}i\'{c}, and S.~Tatikonda, ``Feedback capacity of
  finite-state machine channels,'' \emph{IEEE Trans. Inf. Theory}, vol.~51,
  no.~3, pp. 799--810, Mar. 2005.

\bibitem{TatikondaMitter_IT09}
S.~Tatikonda and S.~Mitter, ``The capacity of channels with feedback,''
  \emph{IEEE Trans. Inf. Theory}, vol.~55, no.~1, pp. 323--349, Jan. 2009.

\bibitem{YangKavcicTatikondaGaussian}
S.~Yang, A.~Kav\u{c}i\'{c}, and S.~C. Tatikonda, ``On the feedback capacity of
  power constrained {G}aussian channels with memory,'' \emph{IEEE Trans. Inf.
  Theory}, vol.~53, no.~3, pp. 929--954, Mar. 2007.

\bibitem{PermuterCuffVanRoyWeissman08}
H.~H. Permuter, P.~Cuff, B.~V. Roy, and T.~Weissman, ``Capacity of the trapdoor
  channel with feedback,'' \emph{IEEE Trans. Inf. Theory}, vol.~54, no.~7, pp.
  3150--3165, Jul. 2009.

\bibitem{Ising_channel}
O.~Elishco and H.~Permuter, ``Capacity and coding for the {I}sing channel with
  feedback,'' \emph{IEEE Trans. Inf. Theory}, vol.~60, no.~9, pp. 5138--5149,
  Sep. 2014.

\bibitem{Bertsekas05}
D.~P. Bertsekas, \emph{Dynamic Programming and Optimal Control: Vols 1 and 2},
  3rd~ed.\hskip 1em plus 0.5em minus 0.4em\relax Belmont, MA.: Athena
  Scientific, 2005.

\bibitem{Massey90}
J.~Massey, ``Causality, feedback and directed information,'' \emph{Proc. Int.
  Symp. Inf. Theory Applic. (ISITA-90)}, pp. 303--305, Nov. 1990.

\bibitem{Kramer03}
G.~Kramer, ``Capacity results for the discrete memoryless network,'' \emph{IEEE
  Trans. Inf. Theory}, vol.~49, no.~1, pp. 4--21, Jan. 2003.

\bibitem{Kim08_feedback_directed}
Y.-H. Kim, ``A coding theorem for a class of stationary channels with
  feedback,'' \emph{IEEE Trans. Inf. Theory.}, vol.~54, no.~4, pp. 1488--1499,
  Apr. 2008.

\bibitem{PermuterWeissmanGoldsmith09}
H.~H. Permuter, T.~Weissman, and A.~J. Goldsmith, ``Finite state channels with
  time-invariant deterministic feedback,'' \emph{IEEE Trans. Inf. Theory},
  vol.~55, no.~2, pp. 644--662, Feb. 2009.

\bibitem{ShraderPermuter09CompoundIT}
B.~Shrader and H.~Permuter, ``Feedback capacity of the compound channel,''
  \emph{IEEE Trans. Inf. Theory}, vol.~55, no.~8, pp. 3629 --3644, Aug. 2009.

\bibitem{Arapos93_average_cose_survey}
A.~Arapostathis, V.~S. Borkar, E.~Fernandez-Gaucherand, M.~K. Ghosh, and
  S.~Marcus, ``Discrete time controlled {M}arkov processes with average cost
  criterion - a survey,'' \emph{SIAM Journal of Control and Optimization},
  vol.~31, no.~2, pp. 282--344, 1993.

\bibitem{horstein_original}
M.~Horstein, ``Sequential transmission using noiseless feedback,'' \emph{IEEE
  Trans. Inf. Theory}, vol.~9, no.~3, pp. 136--143, Jul. 1963.

\bibitem{Kailath_scheme_one}
J.~Schalkwijk and T.~Kailath, ``A coding scheme for additive noise channels
  with feedback--{I}: No bandwidth constraint,'' \emph{IEEE Trans. Inf.
  Theory}, vol.~12, no.~2, pp. 172--182, Apr. 1966.

\bibitem{Kim06_MA}
Y.-H. Kim, ``Feedback capacity of the first-order moving average {G}aussian
  channel,'' \emph{IEEE Trans. Inf. Theory}, vol.~52, no.~7, pp. 3063--3079,
  Jul. 2006.

\bibitem{shayevitz_posterior_mathcing}
O.~Shayevitz and M.~Feder, ``Optimal feedback communication via posterior
  matching,'' \emph{IEEE Trans. Inf. Theory}, vol.~57, no.~3, pp. 1186--1222,
  Mar. 2011.

\end{thebibliography}
\bibliographystyle{IEEEtran}
%
\end{document}